\documentclass[12pt]{iopart}

\usepackage{iopams}
\usepackage[mathscr]{eucal}
\usepackage{harvard}
\usepackage{enumerate}

  \expandafter\let\csname equation*\endcsname\relax
  \expandafter\let\csname endequation*\endcsname\relax

\usepackage{amsmath}

\usepackage{verbatim}
\usepackage{graphicx}
\usepackage{amssymb}
\usepackage{bbm}

\def\squareforqed{\hbox{\rlap{$\sqcap$}$\sqcup$}}
\def\qed{\ifmmode\squareforqed\else{\unskip\nobreak\hfil
\penalty50\hskip1em\null\nobreak\hfil\squareforqed
\parfillskip=0pt\finalhyphendemerits=0\endgraf}\fi}
\def\endenv{\ifmmode\;\else{\unskip\nobreak\hfil
\penalty50\hskip1em\null\nobreak\hfil\;
\parfillskip=0pt\finalhyphendemerits=0\endgraf}\fi}
\newenvironment{proof}{\noindent \textbf{{Proof~}}}{\qed}

\newtheorem{theorem}{Theorem}[section]

\newtheorem{lemma}[theorem]{Lemma}
\newtheorem{corollary}[theorem]{Corollary}

\newtheorem{definition}[theorem]{Definition}

\newcommand{\openone}{\mbox{$1 \hspace{-1.0mm}  {\bf l}$}}

\newcommand*{\eins}{\ensuremath{\mathbbm 1}}

\newcommand{\ud}{\mathrm{d}}

\newcommand*{\bbR}{\mathbb{R}}

\newcommand*{\cE}{\mathcal{E}}

\newcommand*{\cH}{\mathcal{H}}

\newcommand*{\cL}{\mathcal{L}}

\newcommand*{\cP}{\mathcal{P}}

\newcommand*{\cX}{\mathcal{X}}

\newcommand*{\id}{\mathsf{id}}
\newcommand*{\ket}[1]{|#1\rangle}
\newcommand*{\bra}[1]{\langle #1|}
\newcommand*{\braket}[2]{\langle #1| #2 \rangle}

\newcommand*{\proj}[1]{\ket{#1}\bra{#1}}

\newcommand*{\fr}[2]{\frac{#1}{#2}}

\newcommand{\nrm}[1]{\left\|#1\right\|}

\newcommand{\beq}{\begin{equation}}
\newcommand{\eeq}{\end{equation}}

\newcommand*{\assign}{\ensuremath{\kern.5ex\raisebox{.1ex}{\mbox{\rm:}}\kern -.3em =}}

\begin{document}

\title{Quantumness of correlations, quantumness of ensembles and quantum data hiding}

\author{M Piani$^{1,2}$, V Narasimhachar$^3$ and J Calsamiglia$^4$}
\address{${^1}$ Institute for Quantum Computing and Department of Physics and Astronomy, University of Waterloo, Waterloo ON N2L 3G1, Canada\\
$^2$ Department of Physics, University of Strathclyde, Glasgow G4 0NG, UK\\
$^3$ Department of Physics and Astronomy and Institute for Quantum Science and Technology, University of Calgary, Calgary AB T2N 1N4, Canada\\
$^4$ F{\'i}sica Te{\`o}rica: Informaci{\'o} i Fen{\`o}mens Qu{\`a}ntics, Departament de F{\'i}sica, Universitat Aut{\`o}noma de Barcelona, 08193 Bellaterra (Barcelona), Spain}

\date{\today}

\eads{\mailto{mpiani@uwaterloo.ca},\mailto{vnarasim@ucalgary.ca},\mailto{john.calsamiglia@uab.cat}}

\begin{abstract}
We study the quantumness of correlations for ensembles of bi- and multi-partite systems and relate it to the task of quantum data hiding. Quantumness is here intended in the sense of minimum average disturbance under local measurements. We consider a very general framework, but focus on local complete von Neumann measurements as cause of the disturbance, and, later on, on the trace-distance as quantifier of the disturbance.  We discuss connections with entanglement and previously defined notions of quantumness of correlations. We prove that a large class of quantifiers of the quantumness of correlations are entanglement monotones for pure bipartite states. In particular, we define an entanglement of disturbance for pure states, for which we give an analytical expression. Such a measure coincides with negativity and concurrence for the case of two qubits. We compute general bounds on disturbance for both single states and ensembles, and consider several examples, including the uniform Haar ensemble of pure states, and pairs of qubit states. Finally, we show that the notion of ensemble quantumness of correlations is most relevant in quantum data hiding. Indeed, while it is known that entanglement is not  necessary for a good quantum data hiding scheme, we prove that ensemble quantumness of correlations is.
\end{abstract}

\maketitle



\section{Introduction}

Although quantum entanglement~\cite{reviewent} constitutes one of the most counterintuitive aspects of quantum mechanics
and is a key ingredient of quantum information processing~\cite{nielsenchuang}, in recent years other more general quantum features of correlations have attracted a lot of interest.  The role of such general quantumness of correlations has been investigated in areas that go from the foundations of quantum mechanics, to thermodynamics, to quantum computation, to quantum information, to entanglement theory~\cite{reviewdiscord}. 

Discord~\cite{ollivier2001quantum,henderson2001classical}, as well as a number of other related quantifiers~\cite{reviewdiscord}, were introduced to measure such general quantumness. Two fruitful and conceptually interesting approaches to measuring the quantumness of correlations are in terms of disturbance and extraction of correlations. The first approach, disturbance-based, identifies a distributed state as classical if local maximally informative measurements (that is, rank-one projections) exist that do not perturb the state~\cite{luodisturbance}. The second approach, based on the extraction of correlations, identifies a state as classical if local measurements exist that transfer all the correlations present between the quantum subsystems to classical variables/systems~\cite{ollivier2001quantum,piani2008no}. The two approaches are very tightly related, and the two classes of classical states they pin down---those that are not perturbed by local measurement and those whose correlations can be made classical, respectively---coincide~\cite{reviewdiscord}.

In this paper we mostly focus on the quantumness of correlations as understood in terms of measurement-induced disturbance. In most of the recent literature on the quantumness of correlations, a single state distributed among many parties is typically considered. On the other hand, the study of the quantumness of ensembles of states has a long history (see, e.g.,~\cite{ensembleFuchs,horodecki2005common,horodecki2006quantumness}). Recently, the two concepts---the quantumness of correlations and the quantumness of ensembles---have been connected both conceptually and quantitatively. In particular, in~\cite{ensembleLuo2010,ensembleLuo2011,yao2013quantum} an approach was put forward where the quantumness of the ensembles is quantified in terms of the quantumness of correlations of a properly defined bipartite state.

In this paper, on one hand we take a step further and consider the quantumness of correlations of ensembles of distributed states. In particular, we point out how several quantumness measures---either of single-system ensembles, or of the correlations of a single state of a multipartite system, or of the correlations of an ensemble of states of a multipartite system---can be understood in the same formally unified approach.

On the other hand, we consider the role of the quantumness of correlations in quantum data hiding~\cite{hiding2001,divincenzo2002quantum}. In its simplest instance, quantum data hiding consists in the task of encoding a classical bit in a quantum state shared by distant parties---i.e., in letting these parties share one out of two possible quantum states---so that, while such a bit can be recovered perfectly or almost perfectly through a global---i.e., unrestricted---measurement that discriminates between the two states, the bit is almost perfectly hidden from parties that are limited to act via local operations and classical communication (LOCC). It is known that quantum data hiding is possible using pairs of unentangled quantum states~\cite{eggeling2002hiding}. Nonetheless, it is easy to see that some quantumness of correlations must be at play. In the following we prove that, while there are hiding schemes where one of the two hiding states does not display any quantumness of correlations, on one hand (i) the lack of quantumness of correlations in one of the hiding states imposes limits on the hiding scheme and, on the other hand, (ii)  the ensemble composed by the two states in a general hiding scheme must necessarily display a large quantumness of correlations, as quantified by one of the ensemble measures we introduce.

The paper is organized as follows. In Section~\ref{sec:definitions} we look at a general framework to quantify the quantumness of ensembles and the quantumness of correlations of ensembles. We further discuss relations with notions and measures of quantumness of ensembles and correlations already present in the literature, including entanglement. In Section~\ref{sec:bounds} we focus on the particular disturbance induced by complete projective measurements and quantified by the trace distance. We show that this kind of disturbance measure induces a natural entanglement measure on pure bipartite states, for which we provide an analytical expression. We also compute bounds on general disturbance measures and focus on some concrete examples. In Section~\ref{sec:datahiding} we consider the role of the quantumness of correlations in quantum data hiding. We show how the disturbance-based measures of correlations we introduced and studied in the preceding sections provide natural bounds on the quality of quantum data
hiding schemes. We present concluding remarks in Section~\ref{sec:conclusions}.

\section{Measures of ensemble-quantumness}
\label{sec:definitions}

As indicated in the introduction, we focus on the quantumness of ensembles of states as revealed in terms of the (average) disturbance necessarily induced by operations---in particular, measurements---in some restricted class.

\subsection{Definitions}

We adopt the following general definition for disturbance:
\begin{definition}
Given a measure $D[\cdot,\cdot]$ of distance between quantum states, and a set $\cL$ of measurement strategies, we define the \emph{quantumness of an ensemble $\cE\assign\left\{(p_i,\rho^{(i)})\right\}_{i=1}^n$ under $\cL$ as measured by $D$}, or simply \emph{the $(D,\cL)$-quantumness of $\cE$}, as
\begin{equation}
\label{eq:basicdefinition}
Q_{D,\cL}[\cE]\assign\inf_{\Lambda\in\cL}\sum_{i=1}^np_iD\left[\rho^{(i)},\Lambda[\rho^{(i)}]\right],
\end{equation}
where $\Lambda[\rho]$ denotes the resulting state --- typically with classical features (see Definition~\ref{def:classicalstate} below) --- when $\rho$ is subjected to $\Lambda$. Such a notion of quantumness is well-defined only when the operations $\Lambda$ are such that the distance measure $D\left[\rho^{(i)},\Lambda[\rho^{(i)}]\right]$ is well-defined.
\end{definition}

Two meaningful distance measures to consider are the trace distance, $D_1[\sigma,\tau]=1/2\|\sigma-\tau\|_1$, with $\|X\|_1=\Tr\sqrt{X^\dagger X}$, and the relative entropy~\footnote{Here we take the notion of distance in a loose sense, since the relative entropy is not a distance as it is neither symmetric, nor satisfies triangular inequality.},
$S[\sigma\|\tau]=\Tr(\sigma(\log_2\sigma - \log_2 \tau))$, because they have well-understood operational meanings in terms of state distinguishability~\cite{nielsenchuang}. For the sake of simplicity, we will use the notation  $Q_{D_1,\cL}$ and $Q_{S,\cL}$, respectively. 
Another sensible choice would be the Bures distance  $\sqrt{1- F(\sigma,\tau)}$ (or its square), with the fidelity $F(\sigma,\tau)=\Tr\sqrt{\sqrt{\sigma}\tau\sqrt{\sigma}}$~\cite{uhlmann76,jozsa94}, and we will do use it in Section~\ref{sec:datahiding}, but we will mostly focus on the trace distance and the relative entropy. 

As for $\cL$,  in order for our definition of ensemble quantumness to make sense, we must restrict it to sets of operations which admit a meaningful ``post-measurement state'' whence the distance measure is rendered meaningful. In principle, if the states in the ensemble are density matrices acting on a Hilbert space $\cX$, any subset of the set of channels $C(\cX):=\{\Lambda:L(\cX)\rightarrow L(\cX), \Lambda\textrm{ completely positive and trace preserving}\}$, where $L(\cX)$ is the set of operators from $\cX$ to $\cX$, would be a potential mathematically sound choice.  Nonetheless, we aim here at capturing the idea of ``informative'' measurement, either constrained or unconstrained, that necessarily leads to some disturbance for most states~(see, e.g., \cite{d2003heisenberg,kretschmann2008information,luodisturbance} and references therein).  Furthermore, we have in mind the following notions of classicality, for ensembles and for correlations, respectively~\cite{ensembleFuchs,horodecki2005common,horodecki2006quantumness,groisman2007quantumness,piani2008no}.

\begin{definition}
A set of states $\{\rho^{(i)}\}$ is \emph{classical} if all the states in the set commute, i.e. $[\rho^{(i)},\rho^{(j)}]=0$ for all $i,j$, so that the states can be diagonalized simultaneously.
\end{definition}

\begin{definition}
\label{def:classicalstate}
An $n$-partite state $\rho_{A_1A_2\ldots A_n}$ is classical on $A_k$ if $ \rho_{A_1A_2\ldots A_n} = \sum_i p_i \proj{i}_{A_k}\otimes  \bra{i}_{A_k}\rho_{A_1A_2\ldots A_n} \ket{i}_{A_k}$ for some orthonormal basis $\{\ket{i}\}$ of $A_k$. In particular a bipartite state $\rho_{AB}$ is called \emph{classical-quantum} if it is classical on $A$ and \emph{classical-classical} (or \emph{fully classical}) if it is classical on both $A$ and $B$.
\end{definition}
It is worth remarking that there are distributed states that are unentangled (or separable)---i.e., of the form $\sum_i p_i \rho^{(i)}_A\otimes \rho^{(i)}_B$---but not classical~\footnote{Notice that almost all multipartite states are not classical in the sense above~\cite{ferraro2010almost}. This is an additional motivation to study general quantumness quantitatively, rather than qualitatively.}.

For these reasons we mostly focus on complete projective measurements $\Pi\assign\left\{P^{(k)}\right\}_k$ acting as $\Pi[\sigma]=\sum_k P^{(k)} \sigma P^{(k)}$, and on complete local projective measurements $\Pi_A\assign\left\{P_A^{(j)}\right\}_j$ (on subsystem $A$; similarly for subsystem $B$) or complete bilocal ones $\Pi_A\otimes\Pi_B$ (and generalizations thereof for multipartite systems) when we want to focus on the quantumness of correlations. Of course, one could consider  other generalizations, for example to non-complete measurements, or to channels whose Kraus operators have some specified rank~\cite{luo2013hierarchy,sevagprojectors,brodutchdistributedcomputation}, but for the sake of concreteness we will limit ourselves to the explicit cases above.

For the sake of concreteness, we explicitly list the quantifiers of ensemble quantumness corresponding to the set of operations $\{\Pi\}$, $\{\Pi_A\}$ and $\{\Pi_A\otimes\Pi_B\}$, and to the two distance measures mentioned above:
\begin{itemize}
\item ensemble quantumness for single systems:
\begin{align}
Q_{D_1,\{\Pi\}}[\cE]&\assign\min_{\Pi}\sum_{i=1}^np_i\frac{1}{2}\left\|\rho^{(i)}-\Pi[\rho^{(i)}]\right\|_1,\\
Q_{S,\{\Pi\}}[\cE]&\assign\min_{\Pi}\sum_{i=1}^np_iS\left[\rho^{(i)}\right.\left\|\Pi[\rho^{(i)}]\right] = \min_{\Pi}\sum_{i=1}^np_i \left[S(\Pi[\rho^{(i)}]) - S(\rho^{(i)})\right];
\end{align}
\item ensemble quantumness of correlations:
\begin{itemize}
\item one-sided:
\begin{align}
Q_{D_1,\{\Pi_A\}}[\cE]&\assign\min_{\Pi_A}\sum_{i=1}^np_i\frac{1}{2}\left\|\rho^{(i)}-\Pi_A[\rho^{(i)}]\right\|_1,\\
Q_{S,\{\Pi_A\}}[\cE]&\assign\min_{\Pi_A}\sum_{i=1}^np_iS\left[\rho^{(i)}\right.\left\|\Pi_A[\rho^{(i)}]\right]\\
			&\,=\min_{\Pi_A}\sum_{i=1}^np_i \left[S(\Pi[\rho^{(i)}]) - S(\Pi_A[\rho^{(i)}])\right];
\end{align}
\item two-sided:
\begin{align}
Q_{D_1,\{\Pi_A\otimes\Pi_B\}}[\cE]&\assign\min_{\Pi_A\otimes\Pi_B}\sum_{i=1}^np_i\frac{1}{2}\left\|\rho^{(i)}-(\Pi_A\otimes\Pi_B)[\rho^{(i)}]\right\|_1.\\
Q_{S,\{\Pi_A\otimes\Pi_B\}}[\cE]&\assign\min_{\Pi_A\otimes\Pi_B}\sum_{i=1}^np_iS\left[\rho^{(i)}\right.\left\|(\Pi_A\otimes\Pi_B)[\rho^{(i)}]\right]\\
		&\,=\min_{\Pi_A\otimes\Pi_B}\sum_{i=1}^np_i \left[S(\Pi[\rho^{(i)}]) - S(\Pi_A\otimes\Pi_B[\rho^{(i)}])\right].
\end{align}
\end{itemize}
\end{itemize}
 Note that we have used the fact that for projective measurements the infimum in \eqref{eq:basicdefinition} is in fact a minimum. In addition,
for the measures based on relative entropy, we made use of the relation $S(\rho\|\Pi[\rho])=S(\Pi[\rho])-S(\rho)$, valid for any (not necessarily complete) projective measurement $\Pi$, with $S(\sigma):=-\Tr(\sigma\log_2\sigma)$ the von Neumann entropy~\cite{nielsenchuang}.

\subsection{Some basic observations}
\label{sec:basic}

We first remark that the use of one specific distance measure rather than another one in general strongly depends on the context and, potentially, on the convenience of calculation. For example, in the following we will often focus on the quantity $Q_{D_1,\{\Pi_A\}}$, because of its natural connection with the task of discriminating distributed states via LOCC. On the other hand, it is natural to expect $Q_{S,\cL}$ to be more relevant from an information-theoretical point of view (see Section~\ref{sec:relations}). Also, the quantumness measures $Q_{D_1,\cL}$ are always bounded above by unity, and could be considered to be not so helpful in providing insight on the role of quantumness with increasing dimensions of the systems under scrutiny. We do not believe this to be a strong contraindication to the adoption of measures in the class $Q_{D_1,\cL}$; furthermore, one can always reinstate a scaling with dimensions via the composition with the logarithm function. For example, in Section~\ref{sec:bounds} we will find less trivial upper bounds for $Q_{D_1,\{\Pi_A\}}$, and show that, for fixed local dimension $d$, $Q_{D_1,\Pi_A}$ is maximized by maximally entangled states---even in the case of the trivial ensemble made of only one state---assuming in such a case the value $1-1/d$. This suggests to define a logarithmic version of disturbance as
\[
LQ_{D_1,\Pi_A}:=-\log (1-Q_{D_1,\Pi_A}).
\]
Such a quantity varies between $0$ for ensembles of states classical on $A$, to $\log d$ for (ensembles of) maximally entangled states (in dimension $d$).

In general, the properties of the measure $Q_{D,\cL}[\cE]$ will strongly depend on the properties of the distance $D$ and of the class of operations $\cL$. Some basic properties like invariance under unitaries or local unitaries are easily assessed for relevant choices of $D$ and $\cL$. Interestingly, if we suppose that $D$ is jointly convex, as it happens for any norm-based distance---like the trace distance---and for the relative entropy, then $Q_{D,\cL}[\cE]$ is a monotone under coarse graining, independently of the choice of $\cL$. More precisely, if $\cE'=\{(p_i',\rho'^{(i)})\}$ is such that $p'_i\rho'^{(i)}=\sum_{k\in I_i} p_k \rho^{(k)}$ for some starting ensemble $\cE=\left\{(p_j,\rho^{(j)})\right\}_{j=1}^n$ and some partition $\left\{I_i:\cup_i I_i=\{1,2,\ldots, n\}, I_i\cap I_j=\emptyset\,\forall i\neq j\right\}$, then
\[
Q_{D,\cL}[\cE']\leq Q_{D,\cL}[\cE].
\]

A final remark is that through Pinsker's inequality \cite{hiai1981,schumacher2002approximate}, $\|\rho-\sigma\|_1^2\le(\ln4)S[\rho\|\sigma]$, one easily derives the relation 
\begin{equation}
 Q_{D_1,\cL}[\cE]\le\sqrt{\frac{\ln2}{2}}\sqrt{Q_{S,\cL}[\cE]}~\quad\forall \cL,\cE.
\end{equation}

\subsection{Relations with other measures of quantumness and entanglement}
\label{sec:relations}

In this section we relate the family of quantities introduced in Section~\ref{sec:definitions} with quantifiers of quantumness already present in the Literature.

\subsubsection{Quantumness of correlations of a single state}
\label{sec:quantcorrsinglestate}

In the case in which the ensemble $\cE$ is trivial and contains only one state, i.e., $\cE=\{(1,\rho)\}$, the ensemble measures become single-state measures and we write $Q_{D,\cL}[\rho]$ for $Q_{D,\cL}(\{(1,\rho)\})$. In particular the  quantities introduced above trivially vanish in the case we consider single systems and ``global''  complete von Neumann measurements, i.e.,
\begin{align*}
Q_{D_1,\{\Pi\}}[\rho] &= Q_{S,\{\Pi\}}[\rho]=0,
\end{align*}
because one can always consider projective measurements in the eigenbasis of $\rho$. On the other hand, for a distributed state $\rho=\rho_{AB}$, all the following are non-trivial measures of the quantumness of correlations of $\rho_{AB}$:
\begin{align}
Q_{D_1,\{\Pi_A\}}[\rho_{AB}] &= \min_{\Pi_A}\frac{1}{2}\|\rho_{AB} - \Pi_A[\rho_{AB}]\|_1,\label{eq:tracegeomonesided}\\
Q_{S,\{\Pi_A\}}[\rho_{AB}] &=\min_{\Pi_A}S\left[\rho_{AB}\right.\left\|\Pi_A[\rho_{AB}]\right]= \min_{\Pi_A} S(\Pi_A[\rho_{AB}]) - S(\rho_{AB}),\\
Q_{D_1,\{\Pi_A\otimes\Pi_B\}}[\rho_{AB}] &= \min_{\Pi_A\otimes\Pi_B}\frac{1}{2}\|\rho_{AB} - (\Pi_A\otimes\Pi_B)[\rho_{AB}]\|_1,\label{eq:tracegeomtwosided}\\
Q_{S,\{\Pi_A\otimes\Pi_B\}}[\rho_{AB}] &=\min_{\Pi_A\otimes\Pi_B}S\left[\rho_{AB}\right.\left\|(\Pi_A\otimes\Pi_B)[\rho_{AB}]\right]\\
							&= \min_{\Pi_A\otimes\Pi_B} S((\Pi_A\otimes\Pi_B)[\rho_{AB}]) - S(\rho_{AB}).
\end{align}
These four quantifiers correspond to the measurement-induced disturbance~\cite{luodisturbance} due to one-sided or two-sided measurement, measured either entropically or by means of the trace-distance. The latter case, corresponding to a ``trace-norm discord'', has been defined and studied in~\cite{Debarba2012,Rana2013,tracediscordsarandy,nakano2013negativity}; the entropic measures were instead introduced first as quantum deficits~\cite{localnonlocalhorodecki}. Such measures also correspond, either in general (in the entropic case)~\cite{modiunified} or at least in relevant special cases (for the trace-norm)~\cite{tracediscordsarandy,nakano2013negativity}, to the distance of the given state from the set of classical-quantum or classical-classical states. Another interpretation worth mentioning of such measures, again valid either in general or in special cases, is that in terms of entanglement generated in a quantum measurement~\cite{streltsovmeasurement,piani2011all,gharibian2011characterizing,pianiadesso,discordsciarrino}. We refer to~\cite{nakano2013negativity} for a recent and more extensive summary of the relevant properties and relations between these measures.

\subsubsection{Quantumness of single-system ensembles as quantumness of correlations}
\label{sec:quantsinglesystemensemble}

Given a generic ensemble $\cE=\left\{(p_i,\rho_S^{(i)})\right\}_{i=1}^n$ for a system $S$, one can associate with it a bipartite state $\rho_{SX}(\cE)=\sum_i p_i \rho_S^{(i)}\otimes \proj{i}_X$. Using the fact that for relative entropy~\cite{nielsenchuang,cover2012elements,piani2009relative}
\beq
\label{eq:directsumrelent}
S\left[\sum_i p_i \rho_S^{(i)}\otimes \proj{i}_X\right.\left\| \sum_i p_i \sigma_S^{(i)}\otimes \proj{i}_X\right]= \sum_i p_i S\left[\rho_S^{(i)}\right.\left\|  \sigma_S^{(i)}\right],
\eeq
and that the trace norm of a block diagonal matrix is equal to the sum of the trace norms of the blocks, i.e., $\left\|\bigoplus_i M_i \right\|_1 = \sum_i \left\| M_i \right\|_1$~\cite{bhatia1997matrix}, so that 
\beq
\label{eq:directsumtracedist}
\left\|\sum_i p_i \rho_S^{(i)}\otimes \proj{i}_X - \sum_i p_i \sigma_S^{(i)}\otimes \proj{i}_X\right\|_1 = \sum_i p_i \left\| \rho_S^{(i)}-\sigma_S^{(i)}\right\|_1,
\eeq
we have the identities
\begin{align}
Q_{D_1,\{\Pi_S\}}[\rho_{SX}(\cE)]&=Q_{D_1,\{\Pi\}}(\cE)= Q_{D_1,\{\Pi_S\otimes\Pi_X\}}[\rho_{SX}(\cE)],\label{eq:identitiestracenorm}\\
Q_{S,\{\Pi_S\}}[\rho_{SX}(\cE)]&= Q_{S,\{\Pi\}}(\cE)=Q_{S,\{\Pi_S\otimes\Pi_X\}}[\rho_{SX}(\cE)]\label{eq:identitiesrelent},
\end{align}
where the last equality in each of the above equations is due to the fact that $\Pi_X$ can be chosen to project in the basis $\{\ket{i}_X\}$. On the other hand the first equality in each equation holds independently of the class of operations $\cL$ considered, that is, for example,
\[
Q_{D_1,\cL_S}[\rho_{SX}(\cE)]=Q_{D_1,\cL}(\cE), \quad\forall \cL.
\]
The approach consisting in using tools originally introduced to quantify the quantumness of correlations to quantify the quantumness of ensembles was already put forward in, e.g., Refs.~\cite{ensembleLuo2010,ensembleLuo2011,yao2013quantum}, in particular making use of the relative entropy. Here we emphasize that this is a general fact that actually applies to any distance measure that respects a ``direct sum'' rule like~\eqref{eq:directsumrelent} and~\eqref{eq:directsumtracedist}, and fits in a general paradigm as the one we laid out in Section~\ref{sec:definitions}. For example, in the case where $S$ is a composite system itself, that is, $S=AB$, one can have similar relations for the (bipartite) quantumness of correlations of an ensemble of states $\left\{(p_i,\rho_{AB}^{(i)})\right\}$ of $AB$ and the (tripartite) quantumness of correlations of a tripartite state $\rho_{ABX}(\cE)=\sum_ip_i \rho^{(i)}_{AB}\otimes \proj{i}_X$, like
\[
Q_{D_1,\{\Pi_A\otimes\Pi_B\}}[\rho_{ABX}(\cE)]=Q_{D_1,\{\Pi_A\otimes\Pi_B\otimes\Pi_X\}}[\rho_{ABX}(\cE)]=Q_{D_1,\{\Pi_A\otimes\Pi_B\}}\left[\left\{(p_i,\rho_{AB}^{(i)})\right\}\right].
\]

\subsubsection{Quantumness of ensembles for a single mixed state and convex-roof constructions}

In Section~\ref{sec:quantcorrsinglestate} we have seen how our general ensemble quantifiers can reduce to quantifiers for a single state in the case of a trivial ensemble. There are other, less trivial ways to use our ensemble measures when we deal with a single state.

One possibility is that of considering ensemble realizations of that state. For example, we can consider an arbitrary pure-state ensemble $\cE(\rho)=\left\{(p_i,\proj{\psi^{(i)}})\right\}$ such that $\rho=\sum_i p_i \proj{\psi^{(i)}}$. The idea is then that of defining a single-state quantifier of quantum correlations that is based on a minimization over such a decomposition, i.e.,
\[
Q^\textrm{ens}_{D,\cL}(\rho):=\min_{\cE(\rho)} Q_{D,\cL}(\cE(\rho)).
\]
As discussed in Section~\ref{sec:basic}, if $D$ is jointly convex, then $Q_{D,\cL}(\cE)$ is monotonic under coarse-graining of $\cE$, and one has the relation
\[
Q^\textrm{ens}_{D,\cL}(\rho)\geq Q_{D,\cL}(\rho).
\]
We notice that a state that is classical on, let us say, $A$ will admit pure-state ensemble decompositions where all the pure states in the ensemble are classical in the same basis. So, for example, $Q_{D_1,\{\Pi_A\}}(\rho_{AB})=0$ implies $Q^\textrm{ens}_{D_1,\{\Pi_A\}}(\rho_{AB})=0$.

On the other hand, given a mixed state, one can again use pure-state ensemble realizations of that state, but consider the so-called convex-roof construction
\beq
\label{eq:Qconvexroof}
Q^\textrm{cr}_{D,\cL}(\rho):=\min_{\cE(\rho)} \sum_ip_i Q_{D,\cL}(\proj{\psi^{(i)}}).
\eeq
This construction is the standard one used to extend many entanglement measures from bi- and multi-partite states to mixed states~\cite{reviewent}, and indeed $Q^\textrm{cr}_{D,\cL}(\rho)$ is an entanglement measure if $Q_{D,\cL}(\proj{\psi})$ is an entanglement monotone on pure states $\ket{\psi}$~\cite{vidal2000entanglement,michalhorodeckientmeasures}. Note that in the following we often use the shorthand notation $\psi=\proj{\psi}$, hence writing $Q_{D,\cL}(\psi)$.

Notice that the difference between  $Q^\textrm{cr}_{D,\cL}(\rho)$ and  $Q^\textrm{ens}_{D,\cL}(\rho)$, both defined for a single state $\rho$, is that the infimum entering in Definition \eqref{eq:basicdefinition} is or is not, respectively, adapted to each element of the pure-state ensemble of $\rho$. This automatically implies
\beq
\label{eq:ensgeqcr}
Q^\textrm{ens}_{D,\cL}(\rho) \geq Q^\textrm{cr}_{D,\cL}(\rho),
\eeq
independently of the convexity properties of $D$.

\subsubsection{Quantumness and entanglement}

When we say that we are interested in a quantumness of correlations that is more general than entanglement, we have in mind a hierarchy that, above all, is qualitative, but may also be cast in quantitative terms, depending on the choice of quantifiers for entanglement and general quantumness. What we  expect is that the general quantumness of correlations be larger than entanglement also quantitatively. A generic approach that leads to a consistent quantitative hierarchy, i.e., a hierarchy in which the quantumness of correlations is always greater than entanglement, is the one based on the mapping of said quantumness into entanglement itself through a measurement interaction~\cite{streltsovmeasurement,piani2011all,gharibian2011characterizing,pianiadesso,coleshierarchy,nakano2013negativity,discordsciarrino}. Another such hierarchy is the one that naturally arises by considering distance measures from various sets that form a hierarchy themselves~\cite{modiunified}.

In the latter spirit, if one considers the relative entropy of entanglement~\cite{vedral1997quantifying,vedral1998entanglement}
\[
E_R\left[\rho_{AB}\right]:=\min_{\sigma_{AB}\textrm{ separable}} S[\rho_{AB}\|\sigma_{AB}]
\]
it is easy to see that $Q_{S,\{\Pi_A\}}(\rho_{AB})\geq E_R\left[\rho_{AB}\right]$, since $\Pi_A[\rho_{AB}]$ is necessarily separable for $\Pi_A$ a complete von Neumann measurement. On the other hand, $Q_{S,\{\Pi_A\}}$ and the entanglement of formation~\cite{bennett1996mixed}
\[
E_F\left[\rho_{AB}\right] := \min_{\cE(\rho_{AB})}\sum_ip_iS(\Tr_A(\proj{\psi_{AB}^{(i)}}))
\]
do not respect a hierarchy~\cite{luotwoqubits}, i.e., there exist states $\rho_{AB}$ and $\sigma_{AB}$  such that
\[
Q_{S,\{\Pi_A\}} (\rho_{AB})< E_F (\rho_{AB}) \quad\textrm{and}\quad Q_{S,\{\Pi_A\}} (\sigma_{AB})> E_F (\sigma_{AB}).
\]
We observe here that instead
\[
Q^\textrm{ens}_{S,\{\Pi_A\}}(\rho_{AB})\geq E_F\left[\rho_{AB}\right].
\]
Indeed, for a pure state, $Q_{S,\{\Pi_A\}}(\proj{\psi_{AB}})=S\left(\Tr_A(\proj{\psi_{AB}})\right)=E_F(\proj{\psi_{AB}})$~\cite{luodisturbance}, so that $Q^\textrm{cr}_{S,\{\Pi_A\}}(\rho_{AB})=E_F\left[\rho_{AB}\right]$, and we can invoke the general relation \eqref{eq:ensgeqcr}. 

\subsection{Entanglement monotones based on disturbance}

Here we prove that a large class of measures $Q_{D,\cL_A}$, including $Q_{S,\{\Pi_A\}}$ and  $Q_{D,\{\Pi_A\}}$, when restricted to single bipartite pure states (see Section~\ref{sec:quantcorrsinglestate}), are entanglement monotones. By this, we mean that said quantifiers  do not increase \emph{on average} under stochastic LOCC (SLOCC)~\cite{vidal2000entanglement}.

\begin{theorem}
\label{thm:maxquantum}
For any distance measure $D$ that
\begin{enumerate}
\item is invariant under unitaries,
\item is monotonic under general quantum operations (this condition actually comprises  (i)),
\item respects the ``flags'' condition~\footnote{See~\cite{horodecki2005simplifying} for the use of the "flags" condition in entanglement theory.}  $D(\sum_i p_i  \rho_i \otimes \proj{i},\sum_i p_i\sigma_i\otimes \proj{i}) =\sum_i p_i D(\rho_i,\sigma_i)$, for $\{p_i\}$ a probability distribution, $\{\rho_i\}$, $\{\sigma_i\}$ states, and $\{\ket{i}\}$ orthogonal flags,
\end{enumerate}
and for any class $\cL_A$ of local operations $\Lambda_A$  that is closed under conjugation by unitaries, i.e., if $\Lambda_A$ is in $\cL_A$ then also $U_A^\dagger\Lambda_A[U_A \cdot U_A^\dagger]  U_A$ is in $\cL_A$ for all $U_A$, one has that $Q_{D,\cL _A}(\psi_{AB})$ is an entanglement monotone on average, i.e.
 $$
 Q_{D,\cL_A}(\psi_{AB})\geq \sum_i p_i Q_{D,\cL_A}(\phi_{AB}^i),
 $$
 where $\{p_i, \phi_{AB}^i\}$ is a pure-state ensemble obtained from $\psi_{AB}$ by local operations and classical communication.

\end{theorem}
\begin{proof}
We will first note that  $ Q_{D,\cL_A}(\psi_{AB})$  depends only on the Schmidt coefficients $\{p_i\}$ of 
$\ket{\psi_{AB}}=U_A\otimes U_B \ket{\phi_{AB}}$ with $\ket{\phi_{AB}}=\sum\sqrt{p_i}\ket{i i}_{AB}$. This follows from the unitary invariance of $D$ and from  $U_A^\dagger \Lambda_A[U_A \cdot U_A^\dagger ] U_A=\Lambda'_A[ \cdot ] \in\cL_A$ :
\[
\begin{split}
Q_{D,\cL_A}(\psi_{AB}) &=\min_{\Lambda_A\in\cL_A}D(U_A^\dagger \otimes U_B^\dagger \psi_{AB} U_A\otimes U_B,U_A^\dagger \otimes U_B^\dagger \Lambda_A[\psi_{AB}]U_A\otimes U_B)\\
&=\min_{\Lambda_A'\in\cL_A}D({\phi_{AB}},\Lambda'_A[{\phi_{AB}}]))=f(\{p_i\}).
\end{split}
\]
Given the symmetry of the state $\ket{\phi_{AB}}$  if is clear that
\beq
\label{eq:sameonboth}
Q_{D,\cL_A}(\psi_{AB})=\min_{\Lambda_A\in\cL_A}D(\psi_{AB},\Lambda_A[{\psi_{AB}}])=\min_{\Lambda_B\in\cL_B}D(\psi_{AB},\Lambda_B[{\psi_{AB}}])=Q_{D,\cL_B}(\psi_{AB}),
\eeq
i.e., we can consider indifferently a minimization over projections on Alice's or Bob's side. 
In order to prove monotonicity on average, it is sufficient to prove monotonicity under unilocal transformations, i.e. stochastic operations defined through
\beq
\label{eq:unilocal}
\ket{\phi^i_{AB}}=C_A^i\ket{\psi}_{AB}/\sqrt{p_i}\qquad p_i = \Tr( C_A^i\proj{\psi}C_A^{i\dagger}),
\eeq
where the $C_A^i$'s are the Kraus operators of a generic quantum operation on Alice's side. If monotonicity holds under such operations, and the same holds for operations on Bob's side, then monotonicity on average under general LOCC follows, because an LOCC protocol is just a sequence of adaptive unilocal operations~\cite{vidal2000entanglement}. To this end,
\[
\begin{split}
Q_{D,\cL_B}(\psi_{AB})
&=\min_{\Lambda_B\in\cL_B}D(\psi_{AB}, \Lambda_B[{\psi_{AB}}])\\
&\geq \min_{\Lambda_B\in\cL_B}D(\sum_iC^i_A \psi_{AB} C_A^{i\dagger}\otimes\proj{i}_{A'},\sum_iC^i_A\Lambda_B[{\psi_{AB}}]C_A^{i\dagger}\otimes\proj{i}_{A'})\\
&=\min_{\Lambda_B\in\cL_B}D(\sum_i p_i \phi_{i}\otimes\proj{i}_{A'},\sum_i p_i \Lambda_B[{\phi_{i}}]\otimes\proj{i}_{A'})\\
&=\min_{\Lambda_B\in\cL_B}\sum_i p_i D( \phi_{i},\Lambda_B[{\phi_{i}}]))\\
&\geq  \sum_i p_i \min_{\Lambda_B\in\cL_B}D( \phi_{i},\Lambda_B[{\phi_{i}}])\\
&= \sum_i p_i Q_{D,\cL_B}(\phi_{i}).
\end{split}
\]
The first inequality is due to monotonicity of the  distance $D$ under quantum operations, in this case the application of local Kraus operators, with the ``which-operator'' information stored in a classical local flag. Notice that the quantum operation on Alice and the projection on Bob commute. Simply moving the measuring operation to Alice's side thanks to \eqref{eq:sameonboth}, similar steps can be taken in the case of a unilocal operation on Bob's side.

\end{proof}
When particularized to the trace-distance we will call $Q_{D_1,\{\Pi_A\}}$ on pure states the \emph{entanglement of disturbance}, and denote it $E^\textrm{disturbance}$ (see next section).
$Q_{D,\cL_A}(\psi_{AB})$ being an entanglement monotone for pure states, a  straightforward extension to mixed states is provided by its convex roof, $Q^\textrm{cr}_{D_1,\{\Pi_A\}}(\rho_{AB})$,
where, we recall, the superscript $^\textrm{cr}$ indicates that we are taking the convex roof over pure-state ensembles of $\rho_{AB}$, see Eq. \eqref{eq:Qconvexroof}. We remark that Bravyi~\cite{Bravyi2003} also studied and quantified entanglement of pure states as a property that implies non-vanishing disturbance under local measurements, but chose entropy---equivalently, relative entropy---as disturbance quantifier.

It will be important in this work to establish bounds on the quantumness of states and ensembles.
For this purpose it is useful to note that \emph{maximal quantumness corresponds to maximal entanglement}.

\begin{corollary}
For a given fixed dimension of $A$, $d_A$, maximally entangled states are maximally quantum-correlated with respect to any measure $Q_{D,\cL_A}(\rho_{AB})$ that respects the conditions of Theorem~\ref{thm:maxquantum}.
\end{corollary}
\begin{proof}
In \cite{streltsov2012general} it was already proven that a measure of correlations $Q$ that does not increase under operations on at least one-side must be maximal on pure states. It is easy to verify that monotonicity under operations of $D$ implies monotonicity under operations on $B$ for any $Q_{D,\cL_A}$. Furthermore, any pure state of $A$ and $B$ can be obtained via LOCC---in particular, with one-way communication from Bob to Alice---from a maximally entangled state of Schmidt rank $d_A$.
\end{proof}

\section{Trace-norm disturbance: analysis and bounds}
\label{sec:bounds}

Our main goal here is that of finding non-trivial bounds on $Q_{D,\cL}(\mathcal{E})$, focusing on measures like $Q_{D_1,\{\Pi_A\}}$ and $Q_{D_1,\{\Pi_A\otimes\Pi_B\}}$. The latter have the interpretation of quantifiers of the quantumness of correlations in terms of measurement-induced disturbance, where the change is quantified by means of a distance with an operational meaning---the trace distance. This will make such measures key in the connection between ensemble quantumness and quantum data hiding that we will establish in Section~\ref{sec:datahiding}.

To find bounds on ensemble quantumness, we notice that, in general, if $D$ is jointly convex
\beq
\label{eq:generalboundensemble}
\begin{split}
Q_{D,\cL}[\cE]&=\min_{\Lambda\in\cL}\sum_{i=1}^np_iD\left[\rho^{(i)},\Lambda[\rho^{(i)}]\right]\\
&\leq\min_{\Lambda\in\cL}\max_{\ket{\psi}}
D \left[
	\proj{\psi},
	\Lambda[\proj{\psi}]
    \right].
\end{split}
\eeq
For example, we want to calculate
\[
\min_{\Pi}\max_{\ket{\psi}} \frac{1}{2}\| \proj{\psi}-\Pi[\proj{\psi}]\|_1=\max_{\ket{\psi}} \frac{1}{2}\| \proj{\psi}-\bar{\Pi}[\proj{\psi}]\|_1
\]
for the sake of bounding $Q_{D_1,\{\Pi\}}$. Here we got rid of the minimization over the projector, since any change of basis for the projection is irrelevant once we consider the maximization over pure states. Similarly, for the sake of  bounding $Q_{D_1,\{\Pi_A\}}$, it will be enough to calculate $\max_{\ket{\psi_{AB}}} \frac{1}{2}\| \proj{\psi_{AB}}-\bar{\Pi}_A[\proj{\psi_{AB}}]\|_1$.

To begin, it is worth considering with more attention the single-state measures $Q_{D_1,\{\Pi_A\}}$ of Eq.~\eqref{eq:tracegeomonesided} and $Q_{D_1,\{\Pi_A\otimes\Pi_B\}}$ of Eq.~\eqref{eq:tracegeomtwosided}, evaluated on pure states. On one hand, consideration of such measures is interesting in itself; on the other hand, we will see that the tools that we will develop will be useful also to bound ensemble quantumness.

\subsection{Trace-norm disturbance for pure states: Entanglement of disturbance}
\label{sec:entmeasure}

The following lemma will be the key in the study of the maximum disturbances induced by a complete projective measurement on system $A$, for the case of a bipartite pure state.

\begin{lemma}
\label{lem:implicitc}
In the case of a bipartite pure state $\ket{\psi}_{AB}$, the disturbance caused by a one-sided complete projective measurement $\Pi_A$ on $A$, as measured by the trace distance, is given by the positive $c$ such that
\beq
\label{eq:impliciteq}
\sum_i \frac{p_i}{c+p_i} =1,
\eeq
for $p_i=\bra{i}\rho_A\ket{i}$, with $\rho_A=\Tr_B{\proj{\psi}}$, the probability of obtaining outcome $i$ by measuring in the local orthonormal basis $\{\ket{i}\}$.
\end{lemma}
\begin{proof}
Let $\ket{\psi}=\sum_i \ket{i}_A\ket{w_i}_B$, with the vectors $\ket{w_i}$ not necessarily orthogonal and satisfying $\braket{w_i}{w_i}=p_i$. Then $\Pi_A[ \proj{\psi} ] =\sum_i \proj{i}\otimes\proj{w_i}$. We observe that for any vector $\ket v$ and any positive-semidefinite $A$, the positive-definite part of $\proj v-A$ can have rank at most 1, and consequently, if $\proj v-A$ is traceless, then $\|\proj v-A\|_1=2\|\proj v-A\|_\infty$. In our case then,
\[
\frac{1}{2}\|\proj{\psi} - \Pi_A[\proj{\psi}]\|_1 = \|\proj{\psi} - \Pi_A[\proj{\psi}]\|_\infty.
\]
To further simplify things, let us consider a generic normalized vector $\ket{\phi}=\sum_i \ket{i}\ket{z_i}$, where, similarly as before, the vectors $\ket{z_i}$ are not necessarily orthogonal and satisfy $\braket{z_i}{z_i}=q_i$, with $q_i$ the probability of the outcome $i$ when measuring $A$ in the basis $\{\ket{i}\}$. One has
\[
\begin{split}
\|\proj{\psi} - \Pi_A[\proj{\psi}]\|_\infty 	&= \max_{\ket{\phi}} \bra{\phi}(\proj{\psi} - \Pi_A[\proj{\psi}])\ket{\phi}\\
							&= \max_{\ket{\phi}} \left|\sum_i \braket{z_i}{w_i}\right|^2 - \sum_i |\braket{z_i}{w_i}|^2 \\
							&=\max_{\ket{\phi}} \sum_{i> j} \left(\braket{z_i}{w_i}\braket{w_j}{z_j} +\braket{z_j}{w_j}\braket{w_i}{z_i}\right) \\
							&=2\max_{\ket{\phi}} \sum_{i> j}\Re({\braket{z_i}{w_i}\braket{w_j}{z_j}})\\
							&=2\max_{\ket{q}} \sum_{i> j}\sqrt{q_ip_iq_jp_j}\\
							&= \max_{\ket{q}} \left(\sum_i \sqrt{q_ip_i}\right)^2 - \sum_i q_i p_i\\
							&=\max_{\ket{q}} \bra{q} \left(\proj{p} - \sum_i p_i \proj{i} \right)\ket{q},
\end{split}
\]
where we have introduced the notation $\ket{p}=\sum_i \sqrt{p_i} \ket{i}$ (similarly for $\ket{q}$), and used that $\Re({\braket{z_i}{w_i}\braket{w_j}{z_j}})\leq\sqrt{q_ip_iq_jp_j}$, with equality achieved for $\ket{z_i}=(\sqrt{q_i}/\sqrt{p_i}) \ket{w_i}$. Thus, it is sufficient to find the largest (positive) eigenvalue of the operator $\proj{p} - \sum_i p_i \proj{i} $, which we know will be achieved by  $\ket{q}$ of the form $\ket{q}=\sum_i \sqrt{q_i} \ket{i}$. With this ansatz, we set 
\[
\left(\proj{p} - \sum_i p_i \proj{i}\right) \ket{q} = c \ket{q},
\]
from which we find the scalar relations
\[
\braket{p}{q} \sqrt{p_i} - p_i \sqrt{q_i} = c \sqrt{q_i},
\]
i.e.,
\[
\sqrt{q_i} = \braket{p}{q} \frac{\sqrt{p_i}}{c + p_i}.
\]
Imposing consistency for $ \sum_i \sqrt{p_iq_i} =  \braket{p}{q} $,  i.e., imposing
\[
\sum_i\sqrt{p_iq_i} = \sum_i \sqrt{p_i}\left(\braket{p}{q} \frac{\sqrt{p_i}}{c + p_i}\right) = \braket{p}{q} \sum_i\frac{p_i}{c + p_i}\equiv \braket{p}{q},
\]
and noticing that  \eqref{eq:impliciteq} is monotonically decreasing (and hence invertible) for $c\geq 0$, we arrive at the condition of the statement.
\end{proof}
We now make another preliminary observation.
\begin{lemma}
\label{lem:schurconcavity}
Let $\{p_i\}$ indicate a probability distribution. Then 
\[
f_c(\{p_i\}):=\sum_i \frac{p_i}{c + p_i}
\]
is Schur-concave in $\{p_i\}$ for fixed $c\geq 0$. It is also monotonically decreasing for fixed $\{p_i\}$ and for increasing positive $c$. Define also the function
\[
c(\{p_i\}):= \textup{the unique }c\geq0 \textup{ such that } f_c(\{p_i\})=\sum_i \frac{p_i}{c + p_i} =1
\]
on probability vectors. The function $c(\{p_i\})$ is Schur-concave in $\{p_i\}$.
\end{lemma}
\begin{proof}
The Schur-concavity of $f_c$ is a simple consequence of the concavity of $x/(c+x)$ in $x\geq0$ for $c\geq0$, and of the symmetry of $f_c$ in the $p_i$'s. Monotonicity in $c$ is evident.

Consider now a probability distribution $p'_i=\sum_j B_{ij} p_j$ obtained from the probability distribution ${p_i}_i$ by multiplication by a bistochastic matrix $B_{ij}$. Because of the Schur concavity of  $f_c(\{p_i\})$ in $\{p_i\}$ for fixed $c$, we have
\[
\begin{split}
1 &= 	f_{c(\{p_i\})}(\{p_i\})\\
	&=	\sum_j \frac{p_j}{c(\{p_i\}) + p_j}\\
	&\leq  \sum_i \frac{p'_i}{c(\{p_i\}) + p'_i}\\
	& = f_{c(\{p_i\})}(\{p'_i\}).
\end{split}
\] 
Because of the monotonicity of $f_c(\{p_i\})$ in $c$ for fixed $\{p_i\}$, we conclude that $c(\{p'_i\})\geq{c}(\{p_i\})$. This proves that $c(\{p_i\})$ is a Schur-concave function of $\{p_i\}$. 
\end{proof}
Thus we arrive at:
\begin{theorem}
\label{thm:cmonotone}
The one-sided trace-norm disturbance
\[
Q_{D_1,\{\Pi_A\}}(\proj{\psi}) =\min_{\Pi_A} \frac{1}{2} \|\proj{\psi} - \Pi_A[\proj{\psi}]\|_1
\]
is a \emph{bona fide} entanglement monotone for the bipartite pure state $\ket{\psi}=\sum_i \sqrt{p_i}\ket{i}\ket{i}$, here expressed in its Schmidt decomposition. The minimum disturbance is obtained by measuring in the local Schmidt basis and is equal to the positive $c$ such that
\beq
\label{eq:calcentdisturbance}
\sum_i \frac{p_i}{c+p_i} =1.
\eeq
The same holds for the two-sided trace-norm disturbance
\[
Q_{D_1,\{\Pi_A\otimes\Pi_B\}}(\proj{\psi}) = \min_{\Pi_A\otimes\Pi_B} \frac{1}{2} \|\proj{\psi} - (\Pi_A\otimes\Pi_B)[\proj{\psi}]\|_1,
\]
so that we have $Q_{D_1,\{\Pi_A\otimes\Pi_B\}}(\proj{\psi}) = Q_{D_1,\{\Pi_A\}}(\proj{\psi})$ for all $\ket{\psi}_{AB}$.
\end{theorem}
\begin{proof}
From Lemma \ref{lem:implicitc} we have that $\frac{1}{2} \|\proj{\psi} - \Pi_A[\proj{\psi}]\|_1$ is equal to $c(\{q_i\})$ where $q_i$ is the probability of outcome $i$ in the local projective measurement. Let the latter take place in the local Schmidt basis of $\ket{\psi}=\sum_i \sqrt{p_i} \ket{i}\ket{i}$, so that $q_i=p_i$. Let $\{\ket{u_j}\}$ be any other orthonormal basis. The probability of the outcome $j$ in such an alternative basis is
\[
\begin{split}
p'_j	&=\bra{u_j} \rho_A \ket{u_j}\\
	&=\bra{u_j}\left( \sum_i p_i \proj{i} \right) \ket{u_j}\\
	&=\sum_i p_i |\braket{u_j}{i}|^2.
\end{split}
\]
Since the coefficients $B_{ij}= |\braket{u_j}{i}|^2$ form the entries of a bistochastic matrix, the Schur concavity of $c(\{q_i\})$ (Lemma~\ref{lem:schurconcavity}) lets us conclude that the measurement in the Schmidt basis is optimal for the sake of disturbance.

Note that, similarly, the Schur concavity of $c(\{q_i\})$ in $\{q_i\}$ ensures that $Q_{D_1,\{\Pi_A\}}(\proj{\psi}) = c(\{p_i\})$, for $\{\sqrt{p_i}\}$ the Schmidt coefficients of $\ket{\psi}$, is a \emph{bona fide} entanglement measure on pure states~\cite{nielsen1999conditions,vidal2000entanglement} for deterministic LOCC transformations.  Recall that  Theorem~\ref{thm:maxquantum} already shows that  a wide class of disturbance measures, including $Q_{D_1,\{\Pi_A\}}$, are entanglement monotones on pure states,  i.e., they are non-increasing \emph{on average} under (non-deterministic) LOCC. Hence it provides a proof for a   stronger form of monotinicity.

Finally, to see that $Q_{D_1,\{\Pi_A\otimes\Pi_B\}}(\proj{\psi}) = Q_{D_1,\{\Pi_A\}}(\proj{\psi})$, realize that it again holds
\[
\frac{1}{2} \|\proj{\psi} - (\Pi_A\otimes\Pi_B)[\proj{\psi}]\|_1 = \max_{\ket{\phi}} \bra{\phi} \left( \proj{\psi} - (\Pi_A\otimes\Pi_B)[\proj{\psi}] \right) \ket{\phi}.
\]
Consider $\Pi_A$ projecting in the local Schmidt basis and $\ket{\phi}=\sum_i \sqrt{q_i}\ket{i}\ket{i}$ optimal choices for the sake of achieving $\max_{\ket{\phi}} \bra{\phi} \left( \proj{\psi} - \Pi_A[\proj{\psi}] \right) \ket{\phi}=Q_{D_1,\{\Pi_A\}}(\proj{\psi})=|\sum_i\sqrt{p_iq_i}|^2 - \sum_ip_iq_i$ (see the proof Lemma~\ref{lem:implicitc}), in the case of the one-sided measurement. Then, coming back to two-sided projective measurements,
\begin{multline*}
\max_{\ket{\phi}} \bra{\phi} \left( \proj{\psi} - (\Pi_A\otimes\Pi_B)[\proj{\psi}] \right) \ket{\phi}\\
\begin{aligned}
&\geq \bra{\phi} \left( \proj{\psi} - (\Pi_A\otimes\Pi_B)[\proj{\psi}] \right) \ket{\phi}\\
&\geq |\sum_i\sqrt{p_iq_i}|^2 - \sum_i p_i q_i \Tr(\Pi_B[\proj{i}]\Pi_B[\proj{i}]).
\end{aligned}
\end{multline*}
Since $\Tr(\Pi_B[\proj{i}]\Pi_B[\proj{i}])\leq 1$ for all choices of $\Pi_B$, and since $\Pi_A$ was optimal for $Q_{D_1,\{\Pi_A\}}(\proj{\psi})$, we have proven $Q_{D_1,\{\Pi_A\otimes\Pi_B\}}(\proj{\psi}) \geq Q_{D_1,\{\Pi_A\}}(\proj{\psi})$. The last inequality can be saturated for $\Pi_B$ a projection in the local Schmidt basis of $B$.
\end{proof}

We will call $Q_{D_1,\{\Pi_A\}}$ the \emph{entanglement of disturbance} when considering it on pure states, and denote it $E^\textrm{disturbance}(\psi_{AB})$.
Since it is an entanglement monotone on average, it can be naturally extended to an entanglement measure on mixed states by a convex-roof construction~\footnote{This would not  necessarily be true if only deterministic monotonicity was proven.}:
\beq
E^\textrm{disturbance} (\rho_{AB}) := Q^\textrm{cr}_{D_1,\{\Pi_A\}}(\rho_{AB})=\min_{\cE(\rho)} \sum_ip_i E^\textrm{disturbance}(\proj{\psi^{(i)}}).
\eeq

It is instructive to consider some simple cases to get a flavor of the new measure of entanglement, in particular to see how the formula \eqref{eq:calcentdisturbance} plays out. Obviously, in the case of a factorized state $\ket{\alpha}\ket{\beta}$ we have only one non-zero $p_i$, which is equal to 1. So condition \eqref{eq:calcentdisturbance} becomes $1/(c+1) =1$, which is satisfied by $c=0$, as expected. In the case of a maximally entangled state of two qudits, one has $p_i=1/d$ for $i=1,\dots,d$. So, \eqref{eq:calcentdisturbance} becomes $d\frac{1/d}{c+1/d}=1$, which is solved by $c=1-1/d$: this is the maximal value of minimal disturbance due to local projective measurements on pure states of two qudits. For a pure state of two qubits, the probability distribution reads $\{p_i\}=\{p,1-p\}$, and \eqref{eq:calcentdisturbance} becomes
\[
\frac{p}{c+p} + \frac{1-p}{c+1-p}=1
\]
which is satisfied by $c=\sqrt{p(1-p)}$. The latter is the same as the negativity of entanglement~\cite{vidal2001computable} (and, up to a constant factor, concurrence~\cite{wootters1998entanglement}), as to be expected from~\cite{nakano2013negativity,ciccarello2014}.

We conclude this section by providing some explicit upper and lower bounds for $E=E^\textrm{disturbance}$. We recall that the main result of Theorem~\ref{thm:cmonotone} can be restated as the fact that the entanglement of disturbance of $\ket{\psi_{AB}}$ is the positive number $E$ such that
\[
\sum_i \frac{p_i}{E+p_i} =1,
\]
where the $p_i$'s are the Schmidt coefficients of $\ket{\psi_{AB}}$. This is an analytic expression for $E$, as $E$ can simply be considered the inverse of the function $y=f(x)=\sum_i \frac{p_i}{x+p_i}$---with the Schmidt coefficients considered as parameters---evaluated at $y=1$. Nonetheless we provide here some bounds in terms of more standard functions.

In order to find an upper bound to $E$ we can consider the following steps:
\[
\begin{split}
1&=\sum_i \frac{p_i}{E+p_i}\\
&=\frac{p_1}{E+p_1}+(R-1)\sum_{i=2}^R\frac{1}{R-1}\left(\frac{p_i}{E+p_i}\right)\\
&\leq \frac{p_1}{E+p_1}+(R-1)\frac{\sum_{i=2}^R\frac{1}{R-1}p_i}{E+\sum_{i=1}^R\frac{1}{R-1}p_i}\\
&=\frac{p_1}{E+p_1} + \frac{1-p_1}{E+\frac{1-p_1}{R-1}},
\end{split}
\]
where the inequality is due to the concavity of $x/(1+x)$, $p_1$ is the largest probability, and $R$ is the rank (the number of non-vanishing $p_i$'s). From this we find the bound
\beq
\label{eq:upperbound2}
\begin{split}
E&\leq \frac{-2+2 p_1+R-p_1 R+\sqrt{4-4 p_1-4 R+4 p_1^2 R+R^2+2 p_1 R^2-3 p_1^2 R^2}}{2 (-1+R)}\\
&\leq \frac{1}{2} \left(1-p+\sqrt{-3 p^2+1+2 p}\right),
\end{split}
\eeq
where the second bound is obtained from the first in the limit $R\rightarrow \infty$.
On the other hand, to find a lower bound, we can consider
\beq
\label{eq:lowerboundE}
\begin{split}
1&=\sum_i \frac{p_i}{E+p_i}\\
&=\frac{p_1}{E+p_1}+\sum_{i=2}^R \frac{p_i}{E+p_i}\\
&\geq\frac{p_1}{E+p_1}+\sum_{i=2}^R \frac{p_i}{E+p_2}\\
&=\frac{p_1}{E+p_1}+\frac{1-p_1}{E+p_2}.
\end{split} 
\eeq
Here $p_2$ is the second largest probability, hence the inequality. From this, we can find
\[
E\geq \frac{1}{2} \left(1-p_1-p_2+\sqrt{-3 p_1^2+(-1+p_2)^2+2 p_1 (1+p_2)}\right)\geq 1-p_1,
\]
where the rightmost bound is obtained by setting $p_2\equiv p_1$, i.e., loosening the  inequality in \eqref{eq:lowerboundE}.

Both the upper and the lower bound can be checked to be good, in that they converge to the actual value of $E^\textrm{disturbance}$ in both the limit of an unentangled state and a maximally entangled one.

\begin{figure}
\begin{center}
\includegraphics[scale=.6]{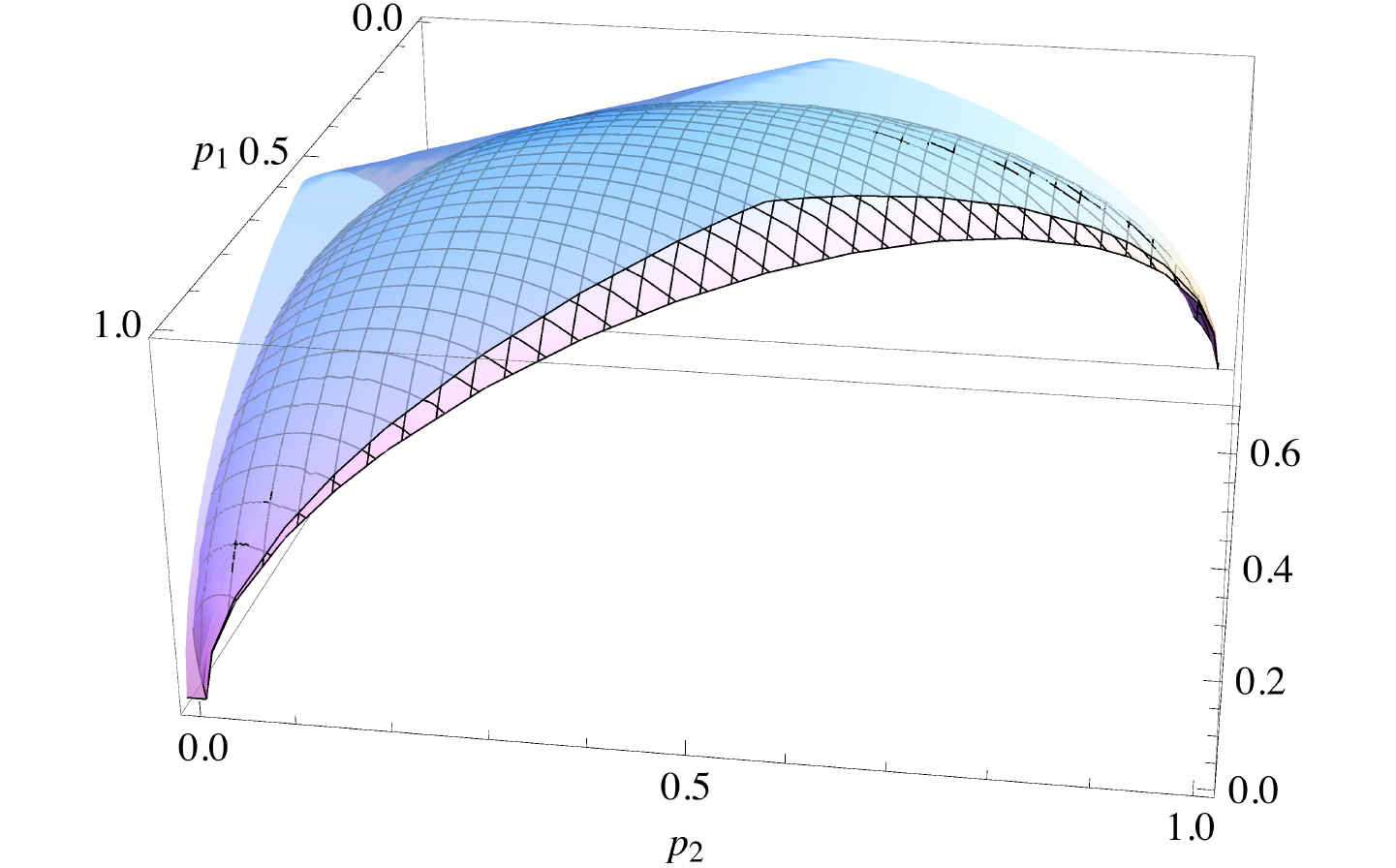}
\caption{Comparison of $E^\textrm{disturbance}$ (lower graph) and the upper bound of Eq.~\eqref{eq:upperbound2} (upper, semitransparent graph) in the case $d=3$ ($d$ is the local dimension). The upper bound is closer to the exact value for highly entangled states corresponding to $p_1,p_2\approx 1/3$ (for which also $p_3=1-p_1-p_2\approx 1/3$).} 
\label{fig:bound}
\end{center}
\end{figure}

\subsection{Bounds on disturbance}

We would like to remark on the difference between calculating (bounds for) the maximal disturbance on \emph{one} distributed state, and on an \emph{ensemble} of distributed states. This is because the measurement in the first case can be tailored to the particular state. More concretely, while
\[
Q_{D_1,\{\Pi_A\otimes\Pi_B\}}(\psi)\leq \max_{\ket{\phi_{AB}}} \min_{\{\Pi_A\otimes\Pi_B\}}\frac{1}{2}\| \proj{\phi}_{AB} - (\Pi_A\otimes\Pi_B) [\proj{\phi}_{AB}] \|_1,
\]
we have (see Eq.~\eqref{eq:generalboundensemble})
\[
Q_{D_1,\{\Pi_A\otimes\Pi_B\}}(\mathcal{E})\leq \min_{\{\Pi_A\otimes\Pi_B\}} \max_{\ket{\phi_{AB}}}\frac{1}{2}\| \proj{\phi}_{AB} - (\Pi_A\otimes\Pi_B) [\proj{\phi}_{AB}] \|_1.
\]
In particular, Theorem~\ref{thm:cmonotone} implies
\beq
\label{eq:boundsingle}
Q_{D_1,\{\Pi_A\otimes\Pi_B\}}(\psi)\leq 1-1/d_A,
\eeq
to be compared with the ensemble bound \eqref{eq:boundensemble} of Corollary~\ref{cor:maxtwowaydist} below, here reported for the convenience of the reader:
\[
Q_{D_1,\{\Pi_A\otimes\Pi_B\}}(\mathcal{E})\leq 1-1/(d_Ad_B).
\]
Nonetheless, as anticipated, in our quest for bounds for the quantumness of ensembles we will be able to take advantage of the results and techniques developed in Section \ref{sec:entmeasure}. 
 
To begin, we remark how in the proof of Lemma~\ref{lem:implicitc} we used a decomposition $\ket{\psi}=\sum_i \ket{i}\ket{w_i}$; this was to analyze projective measurements of the first system in the basis $\{\ket{i}\}$, with the $\ket{w_i}$'s neither orthonormal nor normalized. It is easy to convince oneself that all the calculations we did remain valid in the case of a single system, as long as we interpret the $\ket{w_i}$'s as complex numbers. Therefore, we find

\begin{theorem}
\label{thm:implicitdist}
In the case of a single-system pure state $\psi$, the disturbance caused by a complete projective measurement $\Pi$ in the orthonormal basis $\{\ket{i}\}$, as measured by the trace distance is given by the positive $c$ such that
\beq
\label{eq:implicitedist}
\sum_i \frac{p_i}{c+p_i} =1.
\eeq
for $p_i=|\braket{i}{\psi}|^2$. 
\end{theorem}

We are thus able to find a bound on the quantumness of ensembles, based on the maximum disturbance caused by a fixed projective measurement.

\begin{corollary}
\label{cor:maxdis}
 The maximum disturbance of one state under a von Neumann measurement in dimension $d$, minimized over all measurements and maximized over all states, is given by:
 \beq
 \min_{\{\ket i\}_{i=1}^d}\ \max_{\rho\in\cP_1[\cH_d]}\frac{1}{2}\left\|\rho-\sum_{i=1}^d\left(\proj i\rho\proj i\right)\right\|_1= \ \max_{\rho\in\cP_1[\cH_d]}\frac{1}{2}\left\|\rho-\sum_{i=1}^d\left(\proj i\rho\proj i\right)\right\|_1=1-\fr{1}{d},
 \eeq
 where $\{\ket i\}_{i=1}^d$ denotes a general orthonormal basis spanning the space. Thus,
\beq
\label{eq:boundensemblesingle}
Q_{D_1,\{\Pi\}}(\mathcal{E})\leq 1-\frac{1}{d}.
\eeq
 
\end{corollary}
\begin{proof}
The maximum is attained by a pure state because of the convexity of the trace-norm. From Theorem~\ref{thm:implicitdist} and the Schur concavity of $c$ as a function of the probabilities $\{p_i\}$ (Lemma~\ref{lem:schurconcavity}), it is clear that the maximum is attained for the flat probability distribution $p_i=1/d$, which can be obtained by measuring the the state $(\sum_{i=1}^d \ket{i})/\sqrt{d}$ in the basis $\{\ket{i}\}$.
\end{proof}

Thus, as one may expect, the worst disturbance is obtained by considering a pure state and a projective measurement in a basis that is unbiased with respect to that state, so that $\|\proj{\psi}-\Pi[\proj{\Psi}]\|_1= \|\proj{\psi} - \openone/d\|_1=2(1-1/d)$.

In the bipartite case, the upper bound can be achieved even without entanglement, by local projective measurements acting on appropriately ``skewed'' local pure states. We thus have
\begin{corollary}
\label{cor:maxtwowaydist}
The maximum disturbance of one bipartite state under local complete von Neumann measurements in local dimensions $d_A$ and $d_B$ is given by:
 \beq
\min_{\{\ket{i}_A\}_{i=1}^{d_A},\{\ket{j}_B\}_{i=1}^{d_B}}  \max_{\rho_{AB}}\frac{1}{2}\left\|\rho-\sum_{i=1}^{d_A}\sum_{j=1}^{d_B}\left(\proj{i}_A\otimes\proj{j}_B)\rho_{AB}(\proj{i}_A\otimes\proj{j}_B\right)\right\|_1=1-\fr{1}{d_Ad_B}.
\eeq
Thus,
\beq
\label{eq:boundensemble}
Q_{D_1,\{\Pi_A\otimes\Pi_B\}}(\mathcal{E})\leq 1-\frac{1}{d_A d_B}.
\eeq
\end{corollary}

These results can be generalized to the multipartite case, and
one can cast the bound on disturbance in the following general way.
\begin{theorem}
Consider a composite system $A_1A_2\ldots A_n$, with local dimensions $d_1,d_2,\ldots,d_n$. The maximum disturbance under complete projective measurements on $A_{k_1}A_{k_2}\ldots A_{k_m}$, with $\{k_1,k_2,\ldots,k_m\}\subseteq\{1,2,\ldots,n\}$ is given by:
 \beq
 \label{eq:boundmaxdisturbancemultipartite}
  \min_{\Pi_{A_{k_1}A_{k_2}\ldots A_{k_m}}} \max_{\rho_{A_1A_2\ldots A_n}}\frac{1}{2}
  \left\| \rho_{A_1A_2 \ldots A_n} - \Pi_{A_{k_1} A_{k_2}\ldots A_{k_m}} \left[\rho_{A_1A_2 \ldots A_n}\right] \right\|_1=
  1-\frac{1}{d_{A_{k_1}}d_{A_{k_2}}\cdots d_{A_{k_m}}},
 \eeq
independently of whether the minimization is over arbitrary complete projective measurements on $A_{k_1} A_{k_2}\ldots A_{k_m}$ or over local---with respect to an arbitrary grouping of $A_{k_1}A_{k_2}\ldots A_{k_m}$---ones. 
Thus,
\beq
\label{eq:boundensemblemulti}
Q_{D_1,\{\Pi_{A_{k_1} A_{k_2}\ldots A_{k_m}}\}} (\mathcal{E})\leq  1-\frac{1}{d_{A_{k_1}}d_{A_{k_2}}\cdots d_{A_{k_m}}}.
\eeq
\end{theorem}

We notice that it is possible to tighten all the above bounds on disturbance if one takes also into account the probabilities of the various states in the ensemble. In particular, using a consideration similar to the one used to bound the relative entropy of quantumness of classical-quantum states in~\cite{gharibian2011characterizing}, one can derive the following.
\begin{theorem}
\label{thm:multiimprov}
Consider a composite system $A_1A_2\ldots A_n$, with local dimensions $d_1,d_2,\ldots,d_n$. Suppose $\cE$ is an ensemble comprising, with probability $q$, a state $\rho=\rho_{A_1A_2\ldots A_n}$ which is classical on the individual systems ${A_{k_1} A_{k_2}\ldots A_{k_m}}$; then
\[
Q_{D_1,\Pi_{A_{k_1} A_{k_2}\ldots A_{k_m}}} (\mathcal{E})\leq  (1-q)\left(1-\frac{1}{d_{A_{k_1}}d_{A_{k_2}}\cdots d_{A_{k_m}}}\right),
\]
independently of whether the minimization is over arbitrary complete projective measurements or over local ones. In particular, if there are $n$ states in the ensemble and they all are classical on the individual systems ${A_{k_1} A_{k_2}\ldots A_{k_m}}$, then
\[
Q_{D_1,\Pi_{A_{k_1} A_{k_2}\ldots A_{k_m}}} (\mathcal{E})\leq  \left(1-\frac{1}{n}\right)\left(1-\frac{1}{d_{A_{k_1}}d_{A_{k_2}}\cdots d_{A_{k_m}}}\right).
\]
\end{theorem}
\begin{proof} Consider a projection $\bar\Pi=\bar\Pi_{A_{k_1} A_{k_2}\ldots A_{k_m}}$ that leaves $\rho$ invariant. Without loss of generality, assume $\rho$ is the first state listed in the ensemble. Then
\beq
\begin{split}
 Q_{D_1,\{\Pi_{A_{k_1} A_{k_2}\ldots A_{k_m}}\}}(\cE)
 &=\min_{\{\Pi=\Pi_{A_{k_1} A_{k_2}\ldots A_{k_m}}\}} \sum_{i=1}^n p_i D_1(\rho_i,\Pi[\rho_i])\\
 &\leq q D_1(\rho,\bar\Pi[\rho]) +\sum_{i=2}^n p_i D_1(\rho,\bar\Pi[\rho])\\
& \leq \sum_{i=2}^n p_i \left(1-\frac{1}{d_{A_{k_1}}d_{A_{k_2}}\cdots d_{A_{k_m}}}\right)\\
& = (1-q) \left(1-\frac{1}{d_{A_{k_1}}d_{A_{k_2}}\cdots d_{A_{k_m}}}\right).
\end{split}
\eeq
The first inequality is due to the choice of a particular projection $\bar \Pi$. The second inequality comes from the fact that $\bar\Pi$ is such that $\bar\Pi[\rho]=\rho$, so that $D_1(\rho,\bar\Pi[\rho]) =0$,  and from the general bound~\eqref{eq:boundmaxdisturbancemultipartite}, applied to all the other states in the ensemble.

For the second claim, it suffices to notice that if all $n$ states in the ensemble are classical on the individual systems ${A_{k_1} A_{k_2}\ldots A_{k_m}}$ (possibly in different local orthonormal bases), then at least one of them has associated probability $q\geq 1/n$, because probabilities must sum up to 1.
\end{proof}

It is worth recalling that, if one considers a single system, every given state is classical in its eigenbasis. So, as an application of Theorem~\ref{thm:multiimprov} we find the following improvement over Eq.~\eqref{eq:boundensemblesingle}:
\beq
\label{eq:improvedclassicalquantum}
Q_{D_1,\{\Pi\}}(\mathcal{E})\leq (1-p_\textrm{max})\left(1-\frac{1}{d}\right)\leq\left(1-\frac{1}{n}\right)\left(1-\frac{1}{d}\right),
\eeq
where $d$ is the dimension of the system, $p_{\textrm{max}}$ is the largest among all probabilities with which each state appears in the ensemble, and $n$ is the number of elements in the ensemble. Notice that, because of the identity \eqref{eq:identitiestracenorm}, these bounds on the disturbance of ensembles of a single system can be immediately used to bound the disturbance of correlations of $d \times n$ quantum-classical states $\rho_{SX}$.

\subsection{Examples}

In this section we compute, or provide bounds for, the quantumness of ensembles, both for single systems and for correlations, for some interesting examples. These include qubit ensembles and uniform ensembles of pure states. As we will see, both kinds of examples seem to indicate that the general bound we found are quite good.

\subsubsection{Qubits}

We start by providing a general formula for the single-system ensemble quantumness of ensembles of qubit states.

\begin{theorem}
\label{qubit}
The $(D_1,\{\Pi\})$-quantumness
of an ensemble $\cE\assign\left\{(p_i,\rho^{(i)})\right\}_{i=1}^n$ of qubit states is given by
\begin{equation}
 Q_{D_1,\{\Pi\}}[\cE]=\frac{1}{2}\min_{\hat v\in S^2}\sum_{i=1}^np_i\|\vec r_i\|\left|\sin\left[\angle(\hat v,\vec r_i)\right]\right|,
\end{equation}
where the minimization is performed over all vectors  $\hat v$ on the Bloch sphere, $\vec r_i$ is the Bloch vector corresponding to $\rho^{(i)}$,  $\|\cdot\|\equiv\|\cdot\|_2$ is the Euclidean norm on $\bbR^3$,  $\angle(\hat v,\vec r_i)$ is the angle between the two vectors named.
\end{theorem}
\begin{proof}
 For some qubit state $\rho=\fr{1}{2}\left(\eins+\vec r\cdot\vec\sigma\right)$, if $\Pi$ is the projective measurement along a basis corresponding to the unit vector $\hat v$, then
\begin{align}
 \Pi[\rho]&=\fr{1}{2^3}\left(\eins+\hat v\cdot\vec\sigma\right)\left(\eins+\vec r\cdot\vec\sigma\right)\left(\eins+\hat v\cdot\vec\sigma\right)+\fr{1}{2^3}\left(\eins-\hat v\cdot\vec\sigma\right)\left(\eins+\vec r\cdot\vec\sigma\right)\left(\eins-\hat v\cdot\vec\sigma\right)\nonumber\\
&=\fr{1}{2}\left(\eins+(\hat v\cdot\vec r)\hat v\cdot\vec\sigma\right).
\end{align}
Therefore,
\begin{align}
 \left\|\rho-\Pi[\rho]\right\|_1&=\left\|\fr{1}{2}\left(\eins+\vec r\cdot\vec\sigma-\eins-(\hat v\cdot\vec r)\hat v\cdot\vec\sigma\right)\right\|_1\nonumber\\
&=\fr{1}{2}\left\|\left(\vec r-(\hat v\cdot\vec r)\hat v\right)\cdot\vec\sigma\right\|_1\nonumber\\
&=\left\|\vec r-(\hat v\cdot\vec r)\hat v\right\|\nonumber\\
&=\|\vec r\|\left|\sin\left[\angle(\hat v,\vec r)\right]\right|.
\end{align}
The result follows by the definition of $Q_{D_1,\{\Pi\}}$.
\end{proof}

In the case where the qubit ensemble contains two elements, we are able to provide an explicit analytical formula.

\begin{corollary}
 In the case of an ensemble consisting of two one-qubit states (the case $n=2$ in Theorem \ref{qubit}), the minimum comes out to be
\begin{equation}
 Q_{D_1,\{\Pi\}}\left[\left\{(p,\rho^{(1)}),(1-p,\rho^{(2)})\right\}\right]=\frac{1}{2}\left|\sin\left[\angle(\vec r_1,\vec r_2)\right]\right|\min\left[p\|\vec r_1\|,(1-p)\|\vec r_2\|\right].
\end{equation}
\end{corollary}
\begin{proof}
It is clear that the we should choose $\hat v$ to lie in the plane defined by $\vec r_1$ and $\vec r_2$. This is because, for fixed angle between $\hat v$ and $\vec r_1$, the smallest angle between $\hat v$ and $\vec r_2$ is achieved for $\hat v$ lying in such a plane. Having reduced the problem to a two-dimensional one, we can prove the claim by simply considering Figure~\ref{fig:CQqubit}.

\begin{figure}[htbp]
\begin{center}
\includegraphics[scale=0.5]{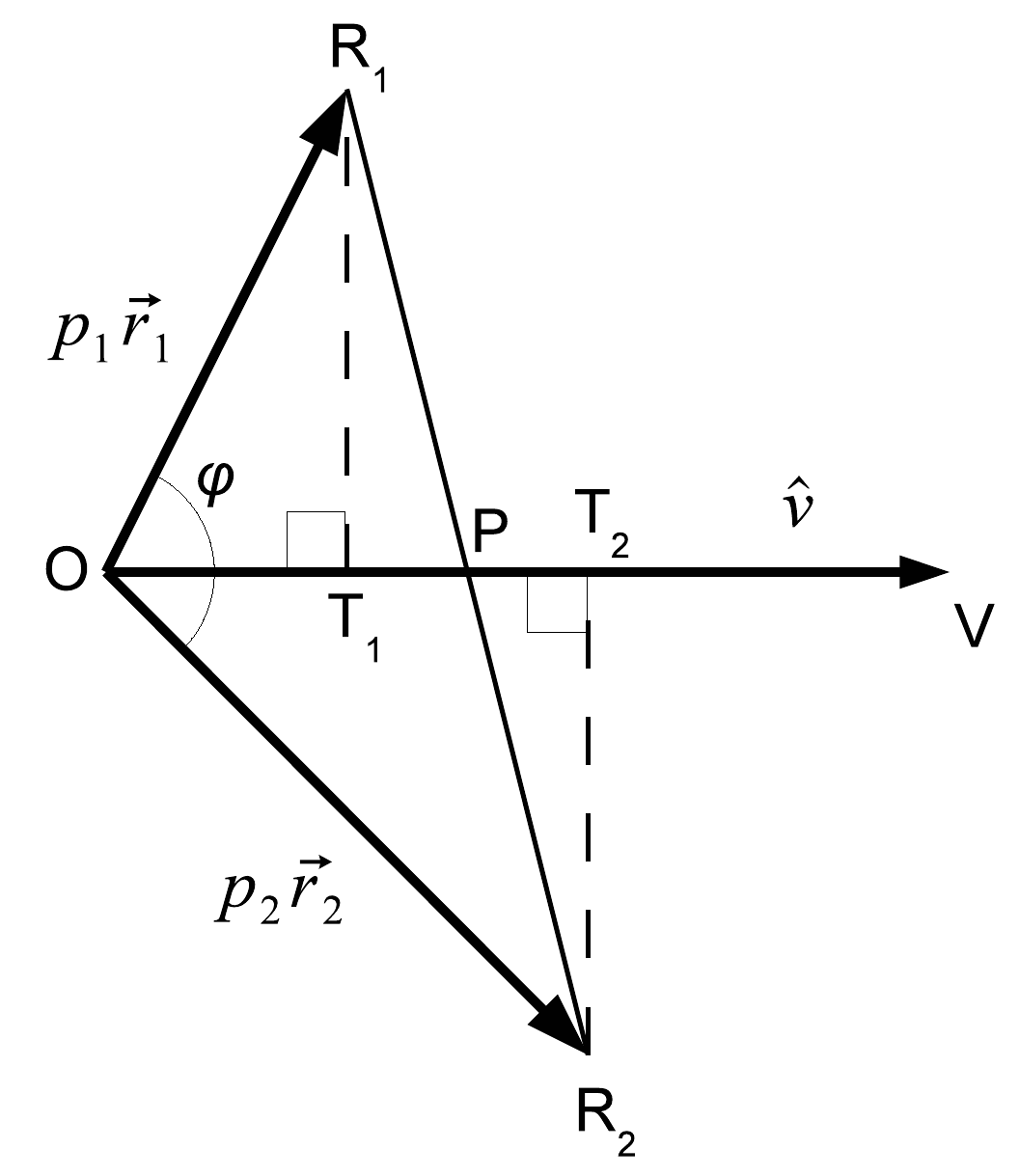}
\caption{Calculation of optimal projective measurement for the least disturbance for an ensemble of two qubit states. $\hat v$ indicates the direction of the projective measurement; $p_i\vec r_i$, $i=1,2$, are the rescaled (by the probability in the ensemble) Bloch vectors of the two states. $\varphi$ is the angle between such rescaled Bloch vectors. The optimal $\hat v$ can always be chosen in the plane defined by the two Bloch vectors, so the problem is a two-dimensional one. The point O corresponds to the centre of the Bloch sphere, and the point P is the point of intersection of $\hat v$ and the segment $R_1R_2$ connecting the endpoints of the two rescaled Bloch vectors.}
\label{fig:CQqubit}
\end{center}
\end{figure}

In terms of the geometric elements present in Figure~\ref{fig:CQqubit}, our objective function can be recast as
\[
\begin{split}
Q_{D_1,\{\Pi\}}\left[\left\{(p,\rho^{(1)}),(1-p,\rho^{(2)})\right\}\right] &= \frac{1}{2}\min_{\hat v\in S^2}\left(p_1\|\vec r_1\|\left|\sin\left[\angle(\hat v,\vec r_1)\right]\right| + p_2\|\vec r_2\|\left|\sin\left[\angle(\hat v,\vec r_1)\right]\right|\right)\\
&=\min_{\hat v\in S^2}\frac{1}{2} (d(R_1,T_1)+d(R_2,T_2)),
\end{split}
\]
where we have used the notation $d(X,Y)$ to denote the Euclidean distance between two points $X$ and $Y$.
We now notice that $\frac{1}{2}d(R_1,T_1)d(O,P)$ is the area of the triangle $OPR_1$. Similarly, $\frac{1}{2}d(R_2,T_2)d(O,P)$ is the area of the triangle $OPR_2$. The sum of the two areas gives the area of the triangle $OR_1R_2$, independently of the position of $P$, i.e. independently of the choice of $\vec v$ in the plane. So
\[
\frac{1}{2}d(R_1,T_1)d(O,P)+\frac{1}{2}d(R_2,T_2)d(O,P) = \frac{1}{2} (d(R_1,T_1)+d(R_2,T_2))d(O,P)=\textrm{const}.
\]
Thus, it is clear that $\hat v$ should be chosen  to maximize $d(O,P)$. This can be done by choosing the measurement axis $\hat v$ parallel to the longest $p_i\vec r_i$. In such a case $\frac{1}{2} (d(R_1,T_1)+d(R_2,T_2))=\frac{1}{2}\left|\sin\left[\angle(\vec r_1,\vec r_2)\right]\right|\min\left[p\|\vec r_1\|,(1-p)\|\vec r_2\|\right]$.
\end{proof}

Via the relations~\eqref{eq:identitiestracenorm}, which equate the quantumness of correlations of a classical-quantum system to the quantumness of a single-system ensemble, we conclude that the classical-quantum state of two qubits that exhibits the largest one-sided quantumness of correlations, as measured by trace-distance disturbance, is, up to local unitaries, the state
\beq
\label{eq:maxCQ}
\frac{1}{2}\proj{0}\otimes\proj{0}+\frac{1}{2}\proj{+}\otimes\proj{1},
\eeq
for which $Q_{D_1,{\Pi_A}} = 1/4$. It is worth remarking that this quantumness matches the upper bound \eqref{eq:improvedclassicalquantum}. The same state~\eqref{eq:maxCQ} is the classical-quantum state exhibiting the largest quantumness of correlations also according to the entropic disturbance~\cite{gharibian2011characterizing}.

\subsubsection{Uniform ensembles}

In this section we will consider uniform ensembles of pure states, that is, ensembles $\mathcal{E}_{\textrm{Haar}}$ of pure states distributed according to the Haar measure~\cite{bengtsson2006geometry}. We will begin with the calculation of the trace-distance single-system ensemble quantumness, $Q_{D_1,\{\Pi\}}$, and move later to the trace-distance ensemble quantumness of correlations, both one-sided, $Q_{D_1,\{\Pi_A\}}$, and two-sided, $Q_{D_1,\{\Pi_A\otimes\Pi_B\}}$.

By symmetry considerations~\footnote{Alternatively, it can be checked.}, it is clear that the choice of basis for the measurement is irrelevant: the average disturbance is going to be the same in all local bases. So
\[
Q_{D_1,\{\Pi\}}(\mathcal{E}_{\textrm{Haar}})=\min_\Pi\int \ud\psi \frac{1}{2}\|\proj{\psi}-\Pi[\proj{\psi}]\|_1=\int \ud\psi \frac{1}{2}\|\proj{\psi}-\Pi[\proj{\psi}]\|_1,
\]
where the projection on the rightmost-hand side is fixed arbitrarily.
We find
\[
\begin{split}
\int \ud\psi \frac{1}{2}\|\proj{\psi}-\Pi[\proj{\psi}]\|_1&\geq \int \ud\psi (1- \bra{\psi}\Pi[\proj{\psi}]\ket{\psi}))\\
&= 1 -\int \ud\psi \Tr(\Pi[\proj{\psi}]^2)\\
&= 1 - \int \ud\psi \Tr\left((\Pi[\proj{\psi}]\otimes\Pi[\proj{\psi}] )W\right)\\
&= 1 - \int \ud\psi \Tr\left((\proj{\psi}\otimes\proj{\psi}) \left((\Pi\otimes\Pi)[V]\right)\right)\\
&= 1 - \Tr\left(\left(\int \ud\psi \proj{\psi}\otimes\proj{\psi}\right) (\Pi\otimes\Pi)[W]\right)\\ 
&= 1 - \Tr\left(\frac{\openone+W}{d(d+1)}(\Pi\otimes\Pi)[W]\right)\\ 
&= 1 -\Tr\left(\frac{\openone+W}{d(d+1)}\sum_i{\proj{i}}\otimes\proj{i}\right)\\
&= 1 -\frac{2d}{d(d+1)}\\
&= 1 -\frac{2}{d+1}.\\
\end{split}
\]
In the above we have introduced the swap operator $W$ acting on two copies of the Hilbert space according to $W\ket{\alpha}\ket{\beta}=\ket{\beta}\ket{\alpha}$. For any arbitrary choice of orthonormal basis $\{\ket{i}\}$, $W$ admits the decomposition $W=\sum_{i,j} \ket{i}\bra{j}\otimes\ket{j}\bra{i}$. We used two useful standard identities involving $W$: $\Tr_{1,2}((X_1\otimes Y_2) W_{1,2})=\Tr(XY)$ and $\int \ud\psi \proj{\psi}\otimes\proj{\psi}=\frac{\openone+W}{d(d+1)}$~\footnote{The first identity can be checked by direct inspection; the second identity can be proved by invoking Schur's lemma in representation theory}. On the other hand,
\[
\begin{split}
Q_{D_1,\{\Pi\}}(\mathcal{E}_{\textrm{Haar}})
&\leq \int \ud\psi \sqrt{1- \bra{\psi}\Pi[\proj{\psi}]\ket{\psi}}\\
&\leq\sqrt{1- \int \ud\psi \bra{\psi}\Pi[\proj{\psi}]\ket{\psi}}\\
&= \sqrt{1 -\frac{2}{d+1}}\\
&=1 - \frac{1}{d+1} + O(d^{-2})
\end{split}
\]
In Corollary~\ref{cor:maxdis} we had found that the greatest trace-distance disturbance, maximized over states, is equal to $1-1/d$ in dimension $d$, and used the same value to bound the (single-system) ensemble disturbance. From the calculations above we see that the ensemble disturbance over the Haar ensemble is essentially the maximal one.

We move now to the one-sided trace-distance quantumness of correlations, $Q_{D_1,\{\Pi_{A}\}}$. We can actually follow several of the steps above to arrive to
\begin{multline*}
Q_{D_1,\{\Pi_A\}}(\mathcal{E}_{\textrm{Haar}})\\
\begin{split}
&\geq 1 - \Tr\left(\frac{\openone_{A_1A_2B_1B_2}+W_{A_1B_1:A_2B_2}}{d_{AB}(d_{AB}+1)}(\Pi_{A_1}\otimes\Pi_{A_2})[W_{A_1B_1:A_2B_2}]\right)\\ 
&= 1 - \Tr\left(\frac{\openone_{A_1A_2B_1B_2}+W_{A_1:A_2}\otimes W_{B_1:B_2}}{d_{AB}(d_{AB}+1)}(\Pi_{A_1}\otimes\Pi_{A_2})[W_{A_1:A_2}\otimes W_{B_1:B_2}]\right)\\ 
&= 1 - \Tr\left(\frac{\openone_{A_1A_2B_1B_2}+W_{A_1:A_2}\otimes W_{B_1:B_2}}{d_{AB}(d_{AB}+1)}\left(\sum_i\proj{i}_{A_1}\otimes\proj{i}_{A_2}\right)\otimes W_{B_1:B_2}\right)\\ 
&=1 - \frac{d_{A}d_B+d_Ad_B^2}{d_{AB}(d_{AB}+1)}\\
&=1-\frac{d_B+1}{d_{A}d_{B}+1}.
\end{split}
\end{multline*}
In the derivation, we have used $d_{AB}=d_A d_B$ and that the swap operator between two copies $A_1B_1$ and $A_2B_2$ of $AB$ can be written as the product of the swaps of the copies of the subsystems, i.e., $W_{A_1B_1:A_2B_2}=W_{A_1:A_2}\otimes W_{B_1:B_2}$. Similarly as before, we also find the upper bound
\[
Q_{D_1,\{\Pi_A\}}(\mathcal{E}_{\textrm{Haar}})\leq \sqrt{1-\frac{d_B+1}{d_{A}d_{B}+1}}
\]
Notice that also this case the average disturbance is comparable with the maximal disturbance for local projective measurements that we computed, $1-1/d_A$.

Finally, it should be clear from the steps in the proof above that, when we consider the two-sided ensemble quantumness of correlations, we go back to the bounds that we obtained for a single system, just with the dimension equal to the total dimension, $d=d_{AB}=d_Ad_B$, obtaining
\[
1-\frac{2}{d_Ad_B+1} \leq Q_{D_1,\{\Pi_A\otimes\Pi_B\}}(\mathcal{E}_{\textrm{Haar}})\leq \sqrt{1-\frac{2}{d_Ad_B+1}}.
\]
Notice that this proves that a bound like \eqref{eq:boundsingle} cannot hold in the case of ensembles, even if it does for single states, and that the dependence on the total dimension of the ensemble bound \eqref{eq:boundensemble} is optimal up to constant factors, at least asymptotically, i.e., for large dimensions.

\section{Quantum data hiding and ensemble quantumness of correlations}
\label{sec:datahiding}

We concern ourselves with hiding classical bits using pairs of quantum states. An $(\epsilon,\delta)$-hiding pair of bipartite states $(\rho_{AB},\sigma_{AB})$ has the properties
$$\fr{1}{2}\left\|\rho_{AB}-\sigma_{AB}\right\|_1=1-\epsilon,$$
$$\fr{1}{2}\left\|\rho_{AB}-\sigma_{AB}\right\|_\mathrm{LOCC}=\delta,$$
where LOCC is meant under the partition $A:B$.  Both the trace distance and the LOCC distance are related to the optimal probability of success in identifying correctly the state, assuming each state is prepared with equal probability, either by global measurements or LOCC measurements~\cite{Matthews2009}:
\begin{align*}
p^\textrm{success}_\textrm{global}&=\frac{1}{2}\left(1+\frac{1}{2}\|\rho_{AB}-\sigma_{AB}\|_1\right)\\
p^\textrm{success}_\textrm{LOCC}&=\frac{1}{2}\left(1+\frac{1}{2}\|\rho_{AB}-\sigma_{AB}\|_\mathrm{LOCC}\right).
\end{align*}
The $\|\cdot\|_\mathrm{LOCC}$ norm is defined on bipartite Hermitian operators as~\cite{Matthews2009}
\[
\|X\|_\mathrm{LOCC}=\max_{\{M_i\}}\sum_i|\Tr(M_iX)|,
\]
where the maximum is taken over all POVMs $\{M_i\}$ that can be realized by LOCC. In the following we will need that
\beq
\label{eq:hierarchymeasurements}
\begin{split}
\fr{1}{2}\left\|\rho_{AB}-\sigma_{AB}\right\|_\mathrm{LOCC}&=\max_{\mathcal{M}_\textrm{LOCC}}\|\mathcal{M}_\textrm{LOCC}[\rho_{AB}-\sigma_{AB}]\|_1\\
&\geq\max_{\mathcal{M}_\textrm{1-LOCC}}\|\mathcal{M}_\textrm{1-LOCC}[\rho_{AB}-\sigma_{AB}]\|_1\\
&\geq\max_{\Pi_A}\|\Pi_A[\rho_{AB}-\sigma_{AB}]\|_1.
\end{split}
\eeq
Here $\mathcal{M}_\textrm{LOCC}[\tau_{AB}]=\sum_i \Tr(M_iX) \proj{i}$ is any LOCC-measurement map~\cite{Matthews2009,piani2009relative}, which comprises LOCC measurements realized by one-way classical communication (1-LOCC), which in turn include measurement schemes where Alice performs a complete projective measurement and communicates the result to Bob. In the latter case, the fact that Bob performs an optimal measurement that depends on the outcome of Alice's projective measurement is automatically taken into account by the definition of the trace norm used in the last line of \eqref{eq:hierarchymeasurements}.

A ``good'' quantum data hiding scheme seeks a pair of states such that both $\epsilon$ and $\delta$ are small positive numbers. The point is that we want to consider a pair of bipartite states that are (almost) perfectly distinguishable by global operations but almost indistinguishable by LOCC. We can define a single parameter for the quality of the hiding scheme in the following way.
\begin{definition}
The \emph{hiding capability} of a pair $(\rho_{AB},\sigma_{AB})$ of states is given by 
\begin{equation}
 \Delta_\mathrm H[\rho_{AB},\sigma_{AB}]\assign\fr{1}{2}\bigg(\left\|\rho_{AB}-\sigma_{AB}\right\|_1-\left\|\rho_{AB}-\sigma_{AB}\right\|_\mathrm{LOCC}\bigg)= 1- \delta - \epsilon.
\end{equation}
\end{definition}
It is clear that $\delta,\epsilon\ll 1$ if and only if $\Delta_\mathrm H[\rho_{AB},\sigma_{AB}]\approx 1$.

It is known that there are good hiding schemes---i.e., with $\Delta_\mathrm H\approx 1$---that do not make use of entanglement~\cite{eggeling2002hiding}.  In the following we will see that, even if entanglement is not strictly needed, some form of quantumness of correlations (in particular, ensemble quantumness) must be present for a hiding scheme to be good. Before getting to such a result, we present evidence that even using one classical state in the pair of hiding states makes the task of quantum data hiding somewhat harder, although not impossible, as the existing constructions present in literature show.

\subsection{On hiding with classical states}
Here we prove a theorem on the limitations of hiding using classical states. For this, we make use of the following observation.
\begin{lemma}
\label{lem:domination}
For any two normalized density operators $\rho$ and $\sigma$ on some Hilbert space $\cH$,
\begin{equation}
\label{eq:statedom}
 \frac{1}{\min\{R(\rho),R(\sigma)\}}F^2(\rho,\sigma)\leq\Tr(\rho\sigma)\le F^2(\rho,\sigma),
\end{equation}
where $F(\rho,\sigma)\assign\Tr\sqrt{\sqrt{\rho}\sigma\sqrt{\rho}}$ is the fidelity of the two states and $R(\rho)$ and $R(\sigma)$ are the ranks of $\rho$ and $\sigma$, respectively.
\end{lemma}
\begin{proof}
We will use the fact that for any matrix $X$ it holds
\beq
\|X\|_2\leq \|X\|_1 \leq \sqrt{R(X)} \|X\|_2
\eeq
with $\|X\|_1=\Tr\sqrt{X^\dagger X}$ the $1$-norm (also known as trace norm) of $X$ and $\|X\|_2=\sqrt{\Tr (X^\dagger X)}$ the $2$-norm (also known as Hilbert-Schmidt norm) of $X$. Then the lemma is proved by simply considering $X=\sqrt{\rho}\sqrt{\sigma}$, since $F(\rho,\sigma)=\|\sqrt{\rho}\sqrt{\sigma}\|_1$ and $\Tr(\rho\sigma)=\|\sqrt{\rho}\sqrt{\sigma}\|_2^2$, and by noting that $R(YZ)\leq\min\{R(Y),R(Z)\}$, while $R(\sqrt{\rho})=R(\rho)$ (similarly for $\sigma$). 
\end{proof}

We are now ready to prove the theorem.
\begin{theorem}
 If a classical-quantum (w.r.t. the partition $A:B$) bipartite state $\rho_{AB}=\sum_i p_i \proj{i}\otimes \rho_i$  is $\epsilon$-distinguishable from another bipartite state $\sigma_{AB}$ under global operations, i.e.,
\beq
\label{eq:hypothmlimit}
\fr{1}{2}\|\rho_{AB}-\sigma_{AB}\|_1\geq 1-\epsilon,
\eeq
and $R=\min\{R(\rho),R(\tilde{\sigma})\}\leq R(\rho)$ is the lesser of the ranks of $\rho$ and $\tilde{\sigma}$, with $\tilde{\sigma}=\sum_i \proj{i}_A \sigma \proj{i}_A$, then $\rho$ is at least $\sqrt{2R \epsilon}$-distinguishable from $\sigma_{AB}$ under $\mathrm{LOCC}$ w.r.t. $A:B$, i.e., 
$$\fr{1}{2}\|\rho_{AB}-\sigma_{AB}\|_\mathrm{LOCC}\ge1-\sqrt{2R \epsilon}.$$
\end{theorem}
\begin{proof}
Besides Lemma~\ref{lem:domination}, we will make use of the well-known relation~\cite{nielsenchuang}
\beq
\label{eq:fidelitytrace}
1-F(\rho,\sigma)\leq D(\rho,\sigma)\leq \sqrt{1-F(\rho,\sigma)^2}.
\eeq
The claim can be proved through the following steps:
\begin{equation}
\begin{split}
\fr{1}{2} \| \rho - \sigma \|_\mathrm{LOCC}
&\stackrel{(i)}{\geq} \fr{1}{2} \| \rho - \tilde{\sigma} \|_1\\
&\stackrel{(ii)}{\geq} 1 - F(\rho,\tilde{\sigma})\\
&\stackrel{(iii)}{\geq} 1 - \sqrt{R}\sqrt{\Tr(\rho\tilde{\sigma})}\\
&\stackrel{(iv)}{=} 1 - \sqrt{R}\sqrt{\Tr(\rho\sigma)}\\
&\stackrel{(v)}{\geq} 1 - \sqrt{R}F(\rho,{\sigma})\\
&\stackrel{(vi)}{\geq} 1 - \sqrt{R}\sqrt{1- D(\rho,\sigma)^2}\\
&\stackrel{(vii)}{\geq} 1 - \sqrt{2R\epsilon}
\label{marco1}
\end{split}
\end{equation}
The steps are justified as follows: $(i)$ holds because a possible LOCC strategy is the one-way communication one with the first step consisting in measuring in the classical basis (for $\rho$) of $A$; $(ii)$ and $(vi)$ are due to Eq.~\eqref{eq:fidelitytrace}; $(iii)$ and $(v)$ hold because of Eq.~\eqref{eq:statedom}; $(iv)$ holds because of the cyclic property of the trace; $(vii)$ holds because of hypothesis Eq.~\eqref{eq:hypothmlimit}.
\end{proof}

Our result points out how, with the use of a classical state,
it is impossible to have a hiding scheme with $\epsilon=0$. Indeed, our bound implies that perfect distinguishability ($\epsilon=0$) by global operations implies perfect distinguishability ($\delta=1$) also by LOCC. On the other hand, it is known that there exist good hiding schemes with perfect distinguishability; they necessarily make use of non-classical states. We will consider such a case in the example below.

We  remark that our result just puts limits on a hiding scheme that uses at least one classical state, but hiding schemes that make use of at least a classical state do exist. For example, Hayden \textit{et al.} \cite{Hayden2004} provide an example of a hiding pair in $\mathbb{C}^d\otimes\mathbb{C}^d$, where one of the states is simply the maximally mixed state, $\openone/d^2$, and the other state is the result of the action of an approximately-randomizing random-unitary map $\mathcal{R}$ on part of a maximally-entangled, i.e., the state $\mathcal{R}\otimes\id[\psi^+]$, with
\[
\mathcal{R}[\rho] = \frac{1}{n}\sum_{i=1}^n U_i\rho U_i^\dagger,\
\]
satisfying
\[
\left\|\rho - \frac{\openone}{d}\right\|_\infty \leq \frac{\eta}{d}.
\]
As proven in~\cite{aubrun2009almost}, improving on~\cite{Hayden2004}, this can be achieved with independent Haar-random unitaries $U_i$ and $n\geq C d/\eta^2$, with $C$ a constant. It is immediate to check that such a scheme achieves $\delta\leq\eta/2$ and $\epsilon\leq n/d^2 \approx C / (d \eta^2)$.

\subsection{Ensemble quantumness bounds the quality of quantum data hiding}

We have seen that the classicality of one of the hiding states does not prevent the hiding scheme from working, although it puts some  ``quality'' constraints on it. With the next theorem we will see that the two states must nonetheless display a large ensemble quantumness of correlations.

\begin{theorem}
\label{HideQ}
 The hiding capability of a pair of bipartite states is bounded above by the \linebreak$(\|\cdot\|_1,\{\Pi_A\})$-quantumness of the ensemble consisting of the two states with equal weights:
\begin{equation}
\label{eq:hidingbound}
 \Delta_\mathrm H[\rho_{AB},\sigma_{AB}]\le 2\, Q_{\|\cdot\|_1,\{\Pi_A\}}\left[\left\{\left(\fr{1}{2},\rho_{AB}\right),\left(\fr{1}{2},\sigma_{AB}\right)\right\}\right].
\end{equation}
\begin{proof}
 By definition,
\begin{align}
\Delta_\mathrm H[\rho_{AB},\sigma_{AB}]=&\fr{1}{2}\bigg(\left\|\rho_{AB}-\sigma_{AB}\right\|_1-\left\|\rho_{AB}-\sigma_{AB}\right\|_\mathrm{LOCC}\bigg)\nonumber\\ =&\fr{1}{2}\left(\left\|\rho_{AB}-\sigma_{AB}\right\|_1-\max_{\mathcal{M}\in\mathrm{LOCC}}\left\|\mathcal{M}[\rho_{AB}]-\mathcal{M}[\sigma_{AB}]\right\|_1\right)\nonumber\\ \le&\fr{1}{2}\left(\left\|\rho_{AB}-\sigma_{AB}\right\|_1-\max_{\Pi_A}\left\|\Pi_A[\rho_{AB}]-\Pi_A[\sigma_{AB}]\right\|_1\right)\nonumber\\
=&\fr{1}{2}\min_{\Pi_A}\bigg(\left\|\rho_{AB}-\sigma_{AB}\right\|_1-\left\|\Pi_A[\rho_{AB}]-\Pi_A[\sigma_{AB}]\right\|_1\bigg)\nonumber\\
\le&\fr{1}{2}\min_{\Pi_A}\bigg(\left\|\rho_{AB}-\Pi_A[\rho_{AB}]\right\|_1\nonumber\\
&~~~~~~~~~~~+\left\|\Pi_A[\rho_{AB}]-\Pi_A[\sigma_{AB}]\right\|_1+\left\|\Pi_A[\sigma_{AB}]-\sigma_{AB}\right\|_1\nonumber\\
&~~~~~~~~~~~~~~~~~~~~~~~~~~~~~~~~~~~~~~~~~~~~~~~~-\left\|\Pi_A[\rho_{AB}]-\Pi_A[\sigma_{AB}]\right\|_1\bigg)\nonumber\\
=&2\,Q_{\|\cdot\|_1,\{\Pi_A\}}\left[\left\{\left(\fr{1}{2},\rho_{AB}\right),\left(\fr{1}{2},\sigma_{AB}\right)\right\}\right].
\end{align}
The first inequality above follows from the fact that one-way LOCC measurements that start with a projective measurement are a subset of all LOCC measurements. The second inequality follows from repeated application of the triangle inequality for the trace norm.
\end{proof}
\label{HidePrTr}
\end{theorem}
Notice that all the steps above could be adapted to the case of projective local measurements on both $A$ and $B$, leading to $ \Delta_\mathrm H[\rho_{AB},\sigma_{AB}]\le Q_{\|\cdot\|_1,\{\Pi_A\otimes \Pi_B\}}\left[\left\{\left(\fr{1}{2},\rho_{AB}\right),\left(\fr{1}{2},\sigma_{AB}\right)\right\}\right]$.

\subsubsection*{Example: Werner hiding pairs}

A well-known hiding pair is that constituted by the two Werner states,
\begin{align}
\sigma_\pm&\assign
\fr{\eins\otimes\eins\pm W}{d(d\pm1)},
\end{align}
with $W$, we recall, the swap operator.
The states $\sigma_+$ and $\sigma_-$ are in fact normalized projectors onto the symmetric and antisymmetric subspaces, respectively. They form a hiding pair, with
$$\fr{1}{2}\nrm{\sigma_+-\sigma_-}_1=1,$$ since they are orthogonal,
and
\beq
\label{eq:wernerhiding}
\fr{1}{2}\nrm{\sigma_+-\sigma_-}_\mathrm{LOCC}=\fr{2}{d+1}.
\eeq
That Eq.~\eqref{eq:wernerhiding} holds comes from the fact that even positive-under-partial-transposition (PPT) measurements, more general than LOCC, do not do better (see~\cite{Matthews2009,eggeling2002hiding,divincenzo2002quantum}), and there are LOCC measurements that achieve the bound (interestingly, the bound is achieved exactly via local orthogonal projections, as easily verified~\cite{privateWatrous}).
Therefore, the hiding gap for this pair is exactly
\beq
\label{eq:gapwerner}
\Delta_\mathrm H[\sigma_+,\sigma_-]=1-\fr{2}{d+1}=\fr{d-1}{d+1}.
\eeq

We calculate the ensemble quantumness of correlations to be
\[
\begin{split}
Q_{\|\cdot\|_1,\{\Pi_A\}}\left[\left\{\left(\fr{1}{2},\sigma_+\right),\left(\fr{1}{2},\sigma_-\right)\right\}\right] &= \min_{\Pi_A}\frac{1}{4}\left[(\|\sigma_+- \Pi_A[\sigma_+]\|_1+\|\sigma_+- \Pi_A[\sigma_+]\|_1\right)\\
&= \min_{\Pi_A}\frac{1}{4}\left( \left\| \fr{\eins\otimes\eins+ W}{d(d+1)}-\Pi_A\left[ \fr{\eins\otimes\eins+ W}{d(d+1)}\right]\right\|_1\right.\\
&\quad\quad \quad \;\; \left.+   \left\| \fr{\eins\otimes\eins- W}{d(d-1)}-\Pi_A\left[ \fr{\eins\otimes\eins-W}{d(d-1)}\right]\right\|_1\right)\\
&= \min_{\Pi_A}\frac{1}{4}\left(\fr{1}{d(d+1)}\left\| W -\Pi_A[W]\right\|_1   +  \fr{1}{d(d-1)}\left\| W -\Pi_A[W]\right\|_1   \right)\\
&=\frac{1}{2(d^2-1)} \left\| W -\sum_i \proj{i}\otimes\proj{i}\right\|_1\\
&=\frac{1}{2(d^2-1)} \left\| \openone-\sum_i \proj{i}\otimes\proj{i}\right\|_1\\
&=\frac{1}{2(d^2-1)} (d^2-d)\\
&=\frac{d}{2(d+1)}
\end{split}
\]

Thus, our bound \eqref{eq:hidingbound} reads $\Delta_H\leq 2\frac{d}{2(d+1)}=\frac{d}{d+1}$ and is quite tight in this case, almost matching the actual quality of the hiding scheme~\eqref{eq:gapwerner}.

\section{Conclusions}
\label{sec:conclusions}

Both the quantumness of ensembles and the quantumness of correlations have been investigated quite intensively in the recent past.  These two notions of quantumness are deeply connected. On one hand, any single-system ensemble of states that exhibits quantumness can be used to construct a distributed state that exhibits some quantumness of correlations~\cite{piani2008no,ensembleLuo2010,ensembleLuo2011,yao2013quantum}. On the other, the study of the quantumness of correlations has often relied on the study of the quantumness of ensembles, intended, e.g., in terms of the impossibility of simultaneously cloning/broadcasting non-commuting single-system states~\cite{barnum1995noncommuting,piani2008no,luo2010decomposition}.

In this paper, in a sense, we combined the notions of quantumness related to dealing with multiple states and the one related to non-classical correlations. We did so by introducing and studying the notion of ensemble quantumness of correlations. Such a notion, we argued, actually fits in a larger unified framework for the study of quantum properties, which encompasses the notion of quantumness of single-system states, as well as of quantumness of correlations of a single bi- or multi-partite state. In our case we chose to depict such a unified framework as based on the quantumness revealed by disturbance under (projective) measurements.

We argued that the ensemble quantumness of correlations plays an important role in one of the basic tasks in quantum information processing: quantum data hiding. Indeed, we noticed how quantum data hiding does not require entanglement, in the sense that there are pairs of hiding states---states used to encoded a bit, so that such bit is recoverable by global quantum operations but not by local operations assisted by classical communication---that are not entangled, and still ensure a good hiding scheme. In this paper we proved that even though quantum data hiding does not require the strong non-classicality linked to entanglement, some non-classicality of correlations must necessarily be present. More precisely, based on existing schemes for quantum data hiding, we argued that the key property is not the quantumness of correlations of the individual states in the hiding pair, as there are hiding schemes that use at least one strictly classical bipartite state. The strictly necessary property seems to rather be the ensemble quantumness of correlation. Indeed, we prove that, if the latter is small, then the quality of the hiding scheme is also necessarily small. This kind of observation is very similar in spirit to the  one that relates the quantumness of correlations to entanglement distribution. For the latter task, Cubitt et al.~\cite{cubitt2003separable} had proven that two parties can increase their entanglement  by exchanging a quantum carrier that is unentangled with both parties. Nonetheless, it was proven in~\cite{streltsov2012quantum,chuan2012quantum} that the increase is bounded by the amount of general non-classical correlations between the particle and the two parties.

After introducing the notion of ensemble quantumness of correlations, we have focused on providing some general bounds on it. This has naturally led us to study in detail the disturbance, as measured by the trace-distance, induced by local projective measurements on a pure bipartite state.  Several researchers interested in the quantumness of correlations had already considered the disturbance induced by local projective measurement on quantum states~\cite{luodisturbance,luo2013hierarchy,nakano2013negativity,tracediscordsarandy}. Nonetheless, as far as we know, we are the first to provide an analytical formula for such disturbance, as measured by trace distance, for all pure states. We actually proved that said disturbance is an entanglement monotone on pure states, in the sense that it is monotonically non-increasing on average under LOCC transformations. This qualifies it to be a good entanglement measure for pure states, which we call \emph{entanglement of disturbance}, and enables a meaningful  extension to mixed states by means of a standard convex roof construction. Returning to ensemble quantumness, we studied several examples, going from the quantumness of ensembles of single-qubit states, to the Haar ensembles for both single systems and bipartite systems. The latter examples perfectly illustrate how the ensemble quantumness of correlations is different from the quantumness of correlations of single states.

The main open problem regards the existence of hiding schemes where both states are individually strictly classical, although obviously (given our resutls) the pair must exhibit ensemble quantumness. Indeed, thanks to existing examples, we know that at least one state in the hiding pair can be strictly classical, but we have proven that its presence poses strong constraints on the hiding scheme. So the question is: Is there a no-go theorem for hiding by means of strictly classical states, even if they are quantum with respect to one another?

\section*{Acknowledgements}

We acknowledge discussions with G. Adesso, A. Brodutch, A. Streltsov, J. Watrous and A. Winter.  We thank D. Reeb for correspondence related to the connection between the results of this paper and~~\cite{jivulescu2014positive} (see the note added below). V.~N. started working on this project during his studies at the Institute for Quantum Computing at the University of Waterloo, and thanks N.~L{\"u}tkenhaus for support and encouragement. J.~C.  is grateful to N.~L{\"u}tkenhaus  and M.~P. for hosting him at the Institute for Quantum Computing, where this work was initiated. M.~P. acknowledges support by NSERC, CIFAR and Ontario Centres of Excellence. J.C. acknowledges support from Spanish MINECO (project FIS2008-01236) with FEDER funds.

\section*{Note added}

After the submission of this manuscript, a work by Jivulescu et al. containing related results appeared~\cite{jivulescu2014positive}.  Jivulescu et al. focus on the problem of whether all bipartite quantum states having a prescribed spectrum satisfy the reduction criterion for separability~\cite{cerf1999reduction,horodecki1999reduction}. In order to attack this problem, Jivulescu et al. evaluate the spectrum of the operator resulting from the action of the reduction map~\cite{cerf1999reduction,horodecki1999reduction} on one party of an arbitrary bipartite pure state, providing an alternative proof of the main formula of Theorem~\ref{thm:cmonotone}. The connection between the present paper and the results of~\cite{jivulescu2014positive} is further discussed in the most recent version of~\cite{jivulescu2014positive}.

\end{document}